\newif\iffullversion
\fullversiontrue

\iffullversion
\documentclass[11pt]{article}
\else
\documentclass{llncs}
\fi

\newcommand{\iffull}[1]{\iffullversion #1 \fi}
\newcommand{\ifconf}[1]{\iffullversion \else #1 \fi}

\iffull{
  \usepackage[letterpaper,margin=2.8cm,centering]{geometry}
  \usepackage{amsthm}

  \newtheorem{theorem}{Theorem}[section]
  
  \newtheorem{definition}[theorem]{Definition}
  \newtheorem{lemma}[theorem]{Lemma}
  
  \newtheorem{corollary}[theorem]{Corollary}

}

\usepackage{amsfonts,amsmath,amsfonts,amsmath,amssymb}
\usepackage{xspace}
\usepackage{algorithm,algorithmic}
\usepackage[textsize=tiny,disable]{todonotes}
\usepackage{framed}
\usepackage{bbm}
\usepackage{enumerate}
\usepackage{epsfig,picins}
\usepackage{dsfont}

\DeclareMathOperator{\E}{E}
\DeclareMathOperator{\Var}{Var}
\DeclareMathOperator{\supp}{supp}
\DeclareMathOperator*{\argmin}{argmin}

\newtheorem{observation}[theorem]{Observation}

\newcommand{\IN}{\ensuremath{\mathbb{N}}}

\newcommand{\ignore}[1]{}
\newcommand{\0}{\mathbb{0}}

\renewcommand{\P}{\ensuremath{\mathcal{P}}}
\newcommand{\D}{\mathcal{D}}
\newcommand{\T}{\ensuremath{\mathcal{T}^P}}
\renewcommand{\S}{\ensuremath{\mathcal{S}}}
\newcommand{\pdls}{\ensuremath{\mbox{pDLS}}\xspace}
\newcommand{\para}[1]{\noindent{\bf #1.}}

\newcommand{\suffix}[2][P]{{#1}_{\succeq{#2}}}
\newcommand{\abs}[1]{\left\lvert{#1}\right\rvert}

\newcommand{\pOne}{P_1}
\newcommand{\pTwo}{P_2}
\newcommand{\late}{L}

\newcommand{\suffChain}{S}
\newcommand{\suffChainFam}{F}
\newcommand{\canonFam}{C}
\newcommand{\suffChainRnd}{\hat{S}}

\newcommand{\td}{\ensuremath{\mbox{TD}}}

{\makeatletter
 \gdef\xxxmark{%
   \expandafter\ifx\csname @mpargs\endcsname\relax 
     \expandafter\ifx\csname @captype\endcsname\relax 
       \marginpar{xxx}
     \else
       xxx 
     \fi
   \else
     xxx 
   \fi}
 \gdef\xxx{\@ifnextchar[\xxx@lab\xxx@nolab}
 \long\gdef\xxx@lab[#1]#2{{\bf [\xxxmark #2 ---{\sc #1}]}}
 \long\gdef\xxx@nolab#1{{\bf [\xxxmark #1]}}
}
\title{
  Approximate Deadline-Scheduling with
  Precedence Constraints
}

\ifconf{
  \author{Hossein Efsandiari\inst{1}, MohammadTaghi Hajiaghyi\inst{1}, Jochen K\"onemann\inst{2}, Hamid Mahini\inst{1}, David Malec\inst{1}, Laura Sanit\`a\inst{2}}

  \institute{University of Maryland, College Park, MD, USA
    \and
    University of Waterloo, Waterloo, Ontario N2L 3G1, Canada}
}
\iffull{
  \author{
    Hossein Efsandiari\footnote{University of Maryland, College Park,
      MD, USA, \{hossein,hajiagha,hmahini,dmalec\}@cs.umd.edu. Supported in part by NSF CAREER award
      1053605, NSF grant CCF-1161626, ONR YIP award N000141110662, a Google Faculty
      Research award, and
      DARPA/AFOSR grant FA9550-12-1-0423.} 
    \and MohammadTaghi Hajiaghyi$^*$
    \and Jochen K\"onemann\footnote{Department of Combinatorics and Optimization,
      University of Waterloo, Waterloo, Ontario, Canada,
      \{jochen,lsanita\}@uwaterloo.ca. Research supported in part by
      the NSRERC Discovery Grant Program} 
    \and Hamid Mahini$^*$ 
    \and David Malec$^*$ 
    \and Laura Sanit\`a
  }
}
  
\begin{document}\sloppy

\abovedisplayskip.10ex
\belowdisplayskip.10ex

 \maketitle

\begin{abstract}
  We consider the classic problem of scheduling a set of $n$ jobs
  non-preemptively on a single machine.  Each job $j$ has non-negative
  processing time, weight, and deadline, and a feasible schedule needs
  to be consistent with {\em chain-like} precedence constraints. The
  goal is to compute a feasible schedule that minimizes the sum of
  penalties of late jobs.  Lenstra and Rinnoy Kan [Annals of
  Disc. Math., 1977] in their seminal work introduced this problem and
  showed that it is strongly NP-hard, even when all processing times
  and weights are $1$.  We study the approximability of the problem
  and our main result is 
  an $O(\log k)$-approximation algorithm for instances with
  $k$ distinct job deadlines.

  We also point out a surprising connection to a model for technology
  diffusion processes in networks that was recently proposed by
  Goldberg and Liu [SODA, 2013]. In an instance of such a problem one
  is given an undirected graph and a non-negative, integer threshold
  $\theta(v)$ for each of its vertices $v$. Vertices $v$ in the graph
  are either {\em active} or {\em inactive}, and an inactive vertex
  $v$ activates whenever it lies in component of size at least
  $\theta(v)$ in the graph induced by itself and all active
  vertices. The goal is now to find a smallest cardinality seed set of
  active vertices that leads to the activation of the entire graph.

  Goldberg and Liu showed that this problem has no
  $o(\log(n))$-approximation algorithms unless NP has quasi-polynomial
  time algorithms, and the authors presented an
  $O(rk\, \log(n))$-approximation algorithm, where $r$ is the radius
  of the given network, and $k$ is the number of distinct vertex
  thresholds. The open question is whether the dependence of the
  approximation guarantee on $r$ and $k$ is avoidable. We answer this
  question affirmatively for instances where the underlying graph is a
  spider. In such instances technology diffusion and precedence
  constrained scheduling problem with unit processing times and
  weights are equivalent problems.

\end{abstract}

\section{Introduction}
\label{sec:intro}

In an instance of the classic {\em precedence-constrained
  single-machine deadline scheduling} problem we are given a
set $[n]:=\{1,\ldots,n\}$ of jobs that need to be scheduled
non-preemptively on a single machine. Each job $j \in [n]$ has a
non-negative deadline $d_j \in \IN$, a processing time $p_j \in \IN$
as well as a non-negative penalty $w_j \in \IN$. A feasible schedule
has to be consistent with precedence constraints that are given
implicitly by a directed acyclic graph $G=([n],E)$; i.e., job $i \in
[n]$ has to be processed before job $j$ if $G$ has a directed
$i,j$-path.  A feasible schedule incurs a penalty of $w_j$ if job $j$
is not completed before its deadline $d_j$. Our goal is then to find a
feasible schedule that minimizes the total penalty of late jobs. In
the standard scheduling notation \cite{GL+79} the problem under
consideration is succinctly encoded as $1|\mbox{prec}|\sum w_jU_j$,
where $U_j$ is a binary variable that takes value $1$ if job $j$ is
late and $0$ otherwise.

Single-machine scheduling with deadline constraints is a
practically important and well-studied subfield of scheduling theory
that we cannot adequately survey here. We refer the reader to Chapter
3 of \cite{Pi12} or Chapter 4 of \cite{BE+07}, and focus here on the
literature that directly relates to our problem. 
The decision version of the single-machine deadline scheduling problem
{\em without} precedence constraints is part of Karp's list of 21
NP-complete problems~\cite{Ka72}, and a fully-polynomial-time
approximation scheme is known~\cite{GL81,Sa76}. The problem becomes
strongly NP-complete in the presence of release dates as was shown by
Lenstra et al.~\cite{LRB77}. Lenstra and Rinnoy Kan~\cite{LR80} later
proved that the above problem is strongly NP-hard even in the special case
where each job has unit processing time and penalty, and the precedence
digraph $G$ is a collection of vertex-disjoint directed paths.

Despite being classical, and well-motivated, 
little is known about the approximability of precedence-constrained
deadline scheduling.  This surprises, given that problems in this
class were introduced in the late 70s, and early 80s, and that these
are rather natural variants of Karp's original 21 NP-hard problems.
The sparsity of results to date suggests that the
combination of precedence constraints and deadlines poses significant
challenges.  We seek to show, however, that these challenges can be
overcome to achieve non-trivial approximations for these important
scheduling problems.  In this paper we focus on the generalization of
the problem studied in~\cite{LR80}, where jobs are allowed to have
arbitrary non-negative processing times, and where we minimize the
{\em weighted} sum of late jobs.  Once more using scheduling notation,
this problem is given by $1|\mbox{chains}|\sum w_j U_j$ (and hereafter
referred to as
\pdls).  Our main result is the following.

\begin{theorem}
\label{thm:lpath}
   \pdls\ has an efficient 
   $O(\log k)$-approximation algorithm, where $k$ is the number of
   distinct job deadlines in the given instance.
\end{theorem}

We note that our algorithm finds a feasible schedule without late jobs
if such a schedule exists.

In order to prove this result, we first introduce a novel, and rather
subtle {\em configuration}-type LP. The LP treats each of the directed
paths in the given precedence system independently. For each path, the
LP has a variable for all nested collections of $k$ suffixes of jobs,
and integral solutions set exactly one of these variables per path to
$1$. This determines which subset of jobs are executed {\em after}
each of the $k$ distinct job deadlines. The LP then has constraints
that limit the total processing time of jobs executed before each of
the $k$ deadlines. While we can show that integral feasible solutions
to our formulation naturally correspond to feasible schedules, the formulation's
integrality gap is large (see \ifconf{the full version \cite{full} for
  details).}\iffull{see Appendix \ref{sec:lp_gap} for details.}
In order to reduce the gap, we strengthen the formulation
using valid inequalities of {\em Knapsack cover}-type
\cite{Ba75,CF+00,HJP75,WO75} (see also \cite{CGK10,KY05}).

The resulting formulation has an exponential number of variables and
constraints, and it is not clear whether it can be solved
efficiently. In the case of chain-like precedences, we are able to
provide an alternate formulation that, instead of variables for nested
collections of suffixes of jobs, has variables for job-suffixes only.
Thereby, we reduce the number of variables to a polynomial of the
input size, while increasing the number of constraints slightly. We do not know
how to efficiently solve even this  alternate LP. However, we are
able to provide a {\em relaxed} separation oracle (in the sense of
\cite{CF+00}) for its constraints, and can therefore use the Ellipsoid
method~\cite{GLS81} to obtain approximate solutions for the alternate LP
of sufficient quality. 

We are able to provide an efficiently computable map between solutions
for the alternate LP, and those of the original exponential-sized
formulation. Crucially, we are able to show that the latter solutions
are structurally nice; i.e., no two nested families of job suffixes in
its support cross! Such {\em cross-free} solutions to the original LP
can then be rounded into high-quality schedules.

Several comments are in order. First, there is a significant
body of research that investigates LP-based techniques for
single-machine, precedence-constrained, minimum weighted
completion-time problems (e.g., see \cite{HSW96,HS+97,Sc96}, and
also \cite{CS05} for a more comprehensive summary of LP-based
algorithms for this problem). None of these LPs seem to be useful for
the objective of minimizing the total penalty of late jobs.  In
particular, converting these LPs requires the introduction of so
called ``big-$M$''-constraints that invariably yield formulations with
large integrality gaps.

Second, using Knapsack-cover inequalities to strengthen an LP
formulation for a given covering problem is not new. 
In the context of approximation algorithms, such inequalities
were used by Carr et al.~\cite{CF+00} in their work on the
Knapsack problem and several generalizations. Subsequently, they also
found application in the development of approximation algorithms for
{\em general covering and packing integer programs}~\cite{KY05}, 
in approximating {\em column-restricted} covering IPs
~\cite{Ko03,CGK10}, as well as in the area of scheduling (without
precedence constraints) \cite{BP10}. Note that our strong formulations for \pdls\ use
variables for (families of) suffixes of jobs in order to encode the
chain-like dependencies between jobs. This leads to formulations that
are not column-restricted, and they also do not fall into the
framework of \cite{KY05} (as, e.g., their dimension is not
polynomial in the input size).

Third, it is not clear how what little work there has been on
precedence-constrained deadline scheduling can be applied to the
problem we study.  The only directly relevant positive result we know
of is that of Ibarra and Kim~\cite{IK78}, who consider the
single-machine scheduling problem in which $n$ jobs need to be
scheduled non-preemptively on a single machine while adhering to
precedence constraints given by {\em acyclic directed forests}, with
the goal to maximize the total profit of jobs completed before a {\em
common deadline $T$}. While the allowed constraints are strictly more
general than the chain-like ones we study, this is more than
outweighed by the fact that all jobs have a common deadline, which
significantly reduces the complexity of the problem and renders it
similar to the well-studied Knapsack problem.  Indeed, we show in
\ifconf{the full version \cite{full} of our paper}
\iffull{Appendix~\ref{sec:single_deadline}}
that \pdls\ with forest precedences
and a single deadline admits a pseudo-polynomial time algorithm as
well.
This implies that the decision version of
\pdls\ is only {\em weakly} NP-complete in this special case. Given
the strong NP-hardness of \pdls\ (as established in \cite{LR80}), it
is unclear how Ibarra and Kim's results can be leveraged for our
problem.

It is natural to ask whether the approximation bound provided in
Theorem \ref{thm:lpath} can be improved. 
In \iffull{Appendix ~\ref{sec:correlated}} \ifconf{\cite{full}}
we provide an example demonstrating that this is
unlikely if we use a path-independent rounding scheme (as in the
proof of Theorem \ref{thm:lpath}). This example highlights that
different paths can play vastly different roles in a solution, and be
critical to ensuring that distinct necessary conditions are met.
Thus, rounding paths independently can lead to many independent
potential points of failure in the process, and significant boosting
of success probabilities must occur if we are to avoid all failures simultaneously.
This means, roughly speaking, that our
analysis is tight and therefore our approximation factor cannot be
improved without significant new techniques.  Given the above, it is
natural to look for dependent rounding schemes for solutions to our
LP. Indeed, such an idea can be made to work for the special case of
\pdls\ with two paths.

\begin{theorem}
\label{thm:2path}
\pdls\ with two paths admits a 2-approximation algorithm
based on a correlated rounding scheme.
\end{theorem}

The proof of Theorem \ref{thm:2path} \ifconf{is deferred to
  \cite{full},}\iffull{is given in Appendix
\ref{sec:correlated},} and shows that the configurational LP used in
the proof of Theorem \ref{thm:lpath} has an integrality gap of at most
$2$ for \pdls\ instances with two paths. This is accomplished using a
randomized rounding scheme that samples families of suffix chains from
the two paths in a correlated fashion instead of independently. The
approach uses the fact that our instances have two paths,
and extending it to general instances appears difficult. 

We point out that the emphasis in Theorem \ref{thm:2path} and its
proof is on the techniques used rather than the approximation
guarantee obtained. In fact, we provide a
dynamic-programming-based exact algorithm for \pdls\ instances with a
fixed number of chains (see
\iffull{Appendix~\ref{sec:constant}}\ifconf{\cite{full}} for details).

\begin{theorem}
  \pdls\ can be solved exactly when the number of chains is fixed.
\end{theorem}

\subsection{Deadline scheduling and technology diffusion}

As we show now, the {\em precedence-constrained single-machine
  deadline scheduling} problem is closely related to the {\em
  technology diffusion} (\td) problem which was recently introduced by
Goldberg and Liu~\cite{GL13} in an effort to model dynamic processes
arising in technology adaptation scenarios.  In an instance of \td, we
are given a graph $G=(V,E)$, and thresholds $\theta(v) \in \{\theta_1,
\ldots, \theta_k\}$ for each $v \in V$. We consider dynamic processes
in which each vertex $v \in V$ is either {\em active} or {\em
  inactive}, and where an inactive vertex $v$ becomes active if, in
the graph induced by it and the active vertices, $v$ lies in a
connected component of size at least $\theta(v)$. The goal in \td\ is
now to find a smallest {\em seed set} $S$ of initially active vertices
that eventually lead to the activation of the entire graph.  Goldberg
and Liu argued that it suffices (albeit at the expense of a constant
factor loss in the approximation ratio) to consider the following
connected abstraction of the problem: find a permutation $\pi=(v_1,
\ldots, v_n)$ of $V$ such that the graph induced by $v_1, \ldots, v_i$ is
connected, for all $i$, and such that
\[ S(\pi) = \{ v_i \,:\, i < \theta(v_i) \} \]
is as small as possible. 

As Goldberg and Liu~\cite{GL13} argue, \td\ has no
$o(\log(n))$-approximation algorithm unless NP has
quasi-polynomial-time algorithms. The authors also presented an $O(rk
\log(n))$-approximation, where $r$ is the diameter of the given graph,
and $k$ is the number of distinct thresholds used in the instance.
K\"onemann, Sadeghian, and Sanit\`a~\cite{KSS13} recently improved
upon this result by presenting a $O(\min\{r,k\}
\log(n))$-approximation algorithm. The immediate open question arising
from \cite{GL13} and \cite{KSS13} is whether the dependence of the
approximation ratio on $r$ and $k$ is avoidable. As it turns out, our
work here provides an affirmative answer for \td\ instances on {\em
  spider} graphs (i.e., trees in which at most one vertex has degree
larger than 2).

\begin{theorem}
  \td\ is NP-hard on spiders. In these graphs, the problem
  also admits an $O(\log(k))$-approximation.
\end{theorem}

The theorem follows from the fact that \td\ in spiders and \pdls\ with
unit processing times, and penalties are equivalent. We sketch the
proof. Given an instance of \td\ on spider $G=(V,E)$, we create a job
for each vertex $v \in V$, and let $d_v=n-\theta(v)+1$, and
$p_v=w_v=1$. We also create a dependence chain for each leg of the
spider; i.e., the job for vertex $v$ depends on all its descendants
in the spider, rooted at its sole vertex of degree larger than $2$. It
is now an easy exercise to see that the \td\ instance has a seed set
of size $s$ iff the \pdls\ instance constructed has a schedule that
makes $s$ jobs late.

\section{Notation}
\label{sec:notation}

In the rest of the paper we will consider an instance of \pdls\ given
by a collection $[n]$ of jobs. Each job $j$ has non-negative
processing time $p_j$, penalty $w_j$ and deadline $d_j$. The precedence
constraints on $[n]$ are induced by a collection of vertex-disjoint,
directed paths $\P=\{P_1, \ldots, P_q\}$. In a feasible schedule job
$j$ has to precede job $j'$ if there is a directed $j,j'$-path in one
of the paths in \P; we will write $j \preceq j'$ to indicate $j$ has
to precede $j'$ from now on for ease of notation, and $j \prec j'$ if
we furthermore have $j \neq j'$. We  denote the set of distinct deadlines in our instance by $\D=\{D_1, \ldots,
D_k\}$, with higher indices corresponding to later deadlines, that is, indexed such that $D_i<D_{i'}$ whenever $i<i'$.
We use the notation $i(j)\in[k]$ to denote the index that the deadline of job $j$ has in the set $\D$, so we have that $d_j=D_{i(j)}$ for all $j\in[n]$.
We say that a job is \emph{postponed} or \emph{deferred} past a certain deadline $D_i$ if the job is executed after $D_i$.
Our goal is to find a feasible schedule that minimizes the
total penalty of late jobs.  Given a directed path $P$, we let
\[\suffix{j}:=\{j' \in [n] \,:\, j \preceq j'\}\]
be the {\em suffix} induced by job $j \in [n]$. We call a sequence 
$\suffChain=(\suffChain_1,\suffChain_2,\dots,\suffChain_k)$ of
suffixes of a given path $P\in \P$ a {\em suffix chain} if
\[P \supseteq \suffChain_1\supseteq \suffChain_2\supseteq \dots \supseteq \suffChain_k;\]
while a suffix chain could have arbitrary length, we will only use suffix chains with length $k=\abs{\D}$.  
Given two suffix
chains $\suffChain$ and $\suffChain'$ with $k$ suffixes each, we say $\suffChain
\preceq \suffChain'$ if $\suffChain_i\supseteq \suffChain'_i$ for all $i\in[k]$.
If we have neither $\suffChain \preceq \suffChain'$ nor $\suffChain' \preceq \suffChain$,
we say that $\suffChain$ and $\suffChain'$ {\em cross}.  
Given two suffix chains
$\suffChain$ and $\suffChain'$, we obtain their {\em join} $\suffChain\vee \suffChain'$ by letting
$(\suffChain\vee \suffChain')_i=\suffChain_i\cup \suffChain'_i$. Similarly, we let the {\em meet} of $\suffChain$
and $\suffChain'$ be obtained by letting $(\suffChain\wedge \suffChain')_i = \suffChain_i \cap \suffChain_i$.

\section{An integer programming formulation}
\label{sec:ex_ip}
Our general approach will be to
formulate the problem as an integer program, to solve its relaxation,
and to randomly round the fractional solution into a feasible schedule
of the desired quality. The IP will have a layered structure. For each
deadline $D_i \in \D$, we want to decide which jobs in $[n]$ are to be
postponed past deadline $D_i$.  We start with the following two easy
but crucial observations.

\begin{observation}\label{obs1}
Consider a path $P \in \P$, and suppose that $j \in P$ is one of the
jobs on this path. If $j$
is postponed past $D_i$ then so are all of $j$'s successors on
$P$. Thus, we may assume w.l.o.g. that the collection of jobs of $P$
that are executed after time $D_i$ forms a {\em suffix} of $P$.
\end{observation}

\begin{observation}\label{obs2}
Consider a path $P \in \P$, and suppose that $j \in P$ is one of the
jobs on this path. If $j$ is postponed past $D_i$, then it is also
postponed past every {\em earlier} deadline $D_{i'}<D_i$. Thus, we may
assume w.l.o.g. that the collections $\suffChain_1,\dots,\suffChain_k$ of jobs of $P$ that are executed
after deadlines $D_1<\dots<D_k$, respectively, exhibit a {\em chain} structure, i.e. $\suffChain_1\supseteq \suffChain_2 \supseteq \dots \supseteq \suffChain_k$.
\end{observation}

Combining the above two observations, we see that for each path
$P\in\P$, the collections of jobs postponed past each deadline form a
suffix chain $\suffChain^P=\suffChain^P_1\supseteq \suffChain^P_2 \supseteq \dots \supseteq
\suffChain^P_k$. In the following we let $\S^P$ denote the collection of suffix
chains for path $P$; we introduce a binary variable $x_{\suffChain}$ for
each suffix chain $\suffChain\in\S^P$ and each $P\in\P$.  In an IP
solution $x_{\suffChain}=1$ for some $\suffChain\in\S^P$ if for each
$i\in[k]$ the set of jobs executed past deadline $D_i$ is precisely
$\suffChain_i$.  We now describe the constraints of the IP in detail.

\para{\eqref{c1} At most one suffix chain of postponed jobs per path}
Since a job can either be deferred or not, and there is no meaningful
way to defer a job twice, we only want to choose at most one suffix
chain per path $P \in \P$. Hence we obtain the constraint:
\begin{equation}\label{c1}\tag{C1}
  \sum_{\suffChain\in{{\S}^P}} x_{\suffChain} \leq 1 \quad \forall P \in \P.
\end{equation}

\para{\eqref{c2} Deferring sufficiently many jobs}
In any feasible schedule, the total processing time of jobs
scheduled before time $D_i$ must be at most $D_i$; conversely, the total
processing time of jobs whose execution is deferred past time $D_i$ must be at least $\Gamma - D_i$, where $\Gamma=\sum_{j\in[n]}p_j$ is the total processing time
of all jobs. This is captured by the following  constraints:

\[\sum_{P \in \P} \sum_{\suffChain \in \S^P} p^{i}_{\suffChain} x_{\suffChain} \geq
\Gamma - D_i \quad \forall i \in [k]
\]
where $p^{i}_{\suffChain}$ is the total processing time of the jobs
contained in $\suffChain_i$.  While the above constraints are certainly
valid, in order to reduce the integrality gap of the formulation and successfully apply our rounding scheme we need to
\emph{strengthen} them, as we now describe.  To this end, suppose that
we are given a chain
\[ \suffChainFam^P = \suffChainFam^P_1 \supseteq \suffChainFam^P_2 \supseteq \suffChainFam^P_3 \supseteq \ldots
\supseteq \suffChainFam^P_k \] of $k$ suffixes of deferred jobs for each path $P
\in \P$, and let $\suffChainFam=\{\suffChainFam^P\}_{P
    \in \P}$ be the family of these suffix chains. Suppose that we knew that we were looking for a schedule
in which the jobs in $\suffChainFam^P_i$ are deferred past deadline $D_i$ for all
$P \in \P$. For each $i \in [k]$, a feasible schedule must now defer
jobs outside $\bigcup_{P \in \P} \suffChainFam^P_i$ of total processing time at
least
\begin{equation}\label{thetais}
\Theta^{i,\suffChainFam} := \max\left\{ (\Gamma - D_i) - \sum_{P \in \P} \sum_{j \in \suffChainFam^P_i}
  p_j, 0 \right\}.
\end{equation}
We obtain the following valid inequality for any feasible schedule:
\begin{equation}\label{c2}\tag{C2}
\sum_{P \in \P} \sum_{\suffChain \in \S^P} p^{i,\suffChainFam}_{\suffChain} x_{\suffChain} \geq
\Theta^{i,\suffChainFam} \quad \forall i \in [k], \, \forall \suffChainFam \in \S,
\end{equation}
where $\S$ is the collection of all families of suffix chains for $\P$ (including the empty family), and
where $p^{i,\suffChainFam}_{\suffChain}$ is the minimum of $\Theta^{i,\suffChainFam}$ and the total processing
time of jobs $j$ that are in $\suffChain_i$ but not in $\suffChainFam^P_i$;
formally,  for $\suffChainFam \in \S$, $i\in [k]$, $P \in \P$, and $\suffChain \in \S^P$, we set
\[
p^{i,\suffChainFam}_{\suffChain} := \min \left\{ \sum_{j\in\suffChain_i\setminus \suffChainFam^P_i} p_{j},
\Theta^{i,\suffChainFam} \right\}.
\]

\eqref{c2} falls into the class of {\em Knapsack Cover}
(KC) inequalities ~\cite{Ba75,CF+00,HJP75,WO75}, and the above {\em
  capping} of coefficients is typical for such inequalities.


All that remains to define the IP is to give the objective function.
Consider a job $j$ on path $P \in \P$, and suppose that the IP
solution $x$ picks suffix chain $\suffChain \in \S^P$. Job $j$ is late
(i.e., its execution ends after time $d_j=D_{i(j)}$) if $j$ is
contained in the suffix $\suffChain_{i(j)}$. We can therefore express the
penalty of suffix chain $\suffChain$ succinctly as
\begin{equation}\label{eqn:alt_cs}
w_{\suffChain} := \sum_{j \in P\,:\, j \in \suffChain_{i(j)}} w_j.
\end{equation}
We can now state the canonical LP relaxation of the IP as follows
\begin{equation} \label{p} \tag{P}
  \min \left\{ \sum_{P\in\P}\sum_{\suffChain \in \S^P} w_{\suffChain}x_{\suffChain} \,:\, \eqref{c1},
  \eqref{c2}, x \geq 0 \right\}.
\end{equation}

For convenience we introduce auxiliary
indicator variables $U_j$ for each job $j \in [n]$. $U_j$ takes value $1$
if $j$'s execution ends after time $d_j$, and hence
\begin{equation}\label{eq:uj}
 U_j := \sum_{{\suffChain}\in{\S^{P}}:j\in\suffChain_{i(j)}}x_{\suffChain},
\end{equation}
where $P$ is the chain containing job $j$.

\section{Rounding the relaxation}
\label{sec:ex_ip_rounding}
\newcommand{\I}[1][]{{\mathcal{I}}_{#1}}  Our rounding scheme does not apply only to
(suitable) feasible points for (P), but in fact allows us to round a
much broader class of (not necessarily feasible) fractional points
$(U,x)$ to integral feasible solutions $(\hat U, \hat x)$ of the
corresponding IP, while only losing a factor of $O(\log k)$ in the
objective value.  As we will see \ifconf{later}\iffull{in Section~\ref{subsec:ellipsoid}}, being able to round this
broader class of points is crucial for our algorithm.  In order
to formally describe the class of points we can round, we need to
introduce the concept of canonical chain families.  Informally, the
canonical suffix chain for a path $P$ defers each job $j\in P$ as much
as possible, subject to ensuring no job in $P$ is deferred past its
deadline. The definition below
makes this formal.

\begin{definition}\label{def:can}
  Given an instance of \pdls, we let $\canonFam_i^P$ be the longest suffix of
  path $P \in \P$ that consists only of jobs whose deadline is
  strictly greater than $D_i$. Jobs in $\canonFam_i^P$ may be scheduled 
  to complete {\em after} $D_i$ without incurring a penalty. 
  We call 
  \[ \canonFam^P := \canonFam^P_1 \supseteq \ldots \supseteq \canonFam^P_k \]
  the {\em canonical suffix chain} for path $P$, and let $\canonFam=\{\canonFam^P\}_{P
    \in \P}$ be the {\em canonical suffix chain family}. 
\end{definition}

Our general approach for rounding a solution $(U,x)$ to program \eqref{p}
is to split jobs into those with large $U_j$ values and those with
small ones.  While we can simply think of ``rounding up''
$U_j$ values when they are already large, we need to utilize the
constraints~\eqref{c1} and~\eqref{c2} to see how to treat jobs with
small $U_j$ values.  As it turns out, in order to
successfully round $(U,x)$ we need it to satisfy 
 the KC-inequality for a {\em single} suffix chain
family only. Naturally this family will depend on the set of
jobs with large $U_j$ value.
We can formalize the above as follows.

Consider any instance $\I$ of \pdls, and let $(U,x)$ be a solution
to (P). Define the set $\late$ of jobs that are late to an extent
of at least $1/(\gamma \log k)$ for a parameter $\gamma > 0$ (whose
value we will make precise at a later point):
\begin{equation*}
  \late = \{ j \,:\,  U_j \geq 1/(\gamma \log k)\}.
\end{equation*}
We now obtain a {\em modified instance of \pdls}, denoted $\I[\late]$,
by increasing the deadline for the jobs in $\late$ to $\Gamma$.  Thus,
jobs in $\late$ can never be late in the modified instance
$\I[\late]$.  Note that since we do not modify the processing time of
any job $j\in[n]$, we have that $p^{i,\suffChainFam}_{\suffChain}$ and
$\Theta^{i,\suffChainFam}$ remain identical in $\I[\late]$ and $\I$ for
all $i$, $\suffChainFam$, and $\suffChain$.  Similarly, each job $j\in[n]$
has the same penalty $w_j$ in $\I$ and $\I[\late]$.  Let $\canonFam$ be the canonical suffix chain
family for $\I[\late]$. We are able to round a solution $(U,x)$ as
long as it satisfies the following conditions:
\begin{itemize}
\item[(a)] for each $P\in\P$, the set $\{\suffChain\in\S^P:x_{\suffChain}>0\}$ is cross-free
\item[(b)] $(U,x)$  is feasible for a relaxation
(P') of \eqref{p} that replaces the constraints~\eqref{c2} by 
\begin{equation}
  \label{c2p}\tag{C2'}
  \sum_{P \in \P} \sum_{\suffChain \in \S^P} p^{i,\canonFam}_{\suffChain} x_{\suffChain} \geq
  \Theta^{i,\canonFam} \quad \forall i \in [k],
\end{equation}
where $\canonFam$ is the canonical suffix chain family for the modified \pdls instance $\I[\late]$.
\end{itemize}
In the next section, we see how we can find solutions satisfying both
of these conditions.

Suppose $(U,x)$ is a solution to \eqref{p} that satisfies (a) and
(b). Obtain $x^0$ by letting $x^0_{\suffChain}=x_{\suffChain}$ if $\suffChain$
makes at least one job $j \in [n]$ late in $\I[\late]$, and let
$x^0_{\suffChain}=0$ otherwise. Define $U^0 \leq U$ as in
\eqref{eq:uj} (with $x^0$ in place of $x$), and note that
$(U^0,x^0)$ satisfies (a) and (b). 
Let us now round $(U^0,x^0)$. We
focus on path $P \in \P$, and define the support of $(U^0,x^0)$
induced by $P$:
\[
\T := \{\suffChain\in\S^P: x^0_{\suffChain}>0 \}. 
\]
As this set is cross-free by assumption (a), \T\ has a well-defined
maximal element $\suffChain^\ast$ with $\suffChain \preceq
\suffChain^\ast$ for all $\suffChain \in \T$ (recall, $\suffChain
\preceq \suffChain^\ast$ means $\suffChain$ defers no less jobs
past every deadline $D_i$ than $\suffChain^\ast$ does).  By
definition, $\suffChain^\ast$ makes at least one job $j \in
[n]\setminus \late$ late. Since $\suffChain^\ast$ is maximal in \T\ it
therefore follows that $j$ is late in {\em all} $\suffChain \in \T$.
Using the definition of $(U^0,x^0)$ as well as the fact that $j
\not\in \late$ we obtain
\begin{equation}\label{eq:1}
  \sum_{\suffChain\in\T}x^0_{\suffChain}
  =\sum_{\suffChain\in\S^P:j \in \suffChain_{i(j)}}x^0_{\suffChain}
  = U^0_{j} \leq U_j <
  \frac{1}{\gamma \log k}.
\end{equation}
We let $(\bar{U},\bar{x}) = \gamma\log k \cdot (U^0,x^0)$ and obtain the
following lemma. 
\ifconf{The proof of this and several of the following lemmas are
  deferred to \cite{full} because of space limitations.}

\begin{lemma} \label{lem:boostfeas}
  $(\bar{U}, \bar{x})$ satisfies 
  \begin{align} 
    \sum_{\suffChain\in \S^P} \bar{x}_{\suffChain} & \leq
    1 \label{c1S}\tag{$\overline{\mbox{C1}}$} \quad \forall i \in [k], \forall P
    \in \P\\
    \sum_{P \in \P} \sum_{\suffChain\in\S^P}  p^{i,\canonFam}_{\suffChain}\bar{x}_{\suffChain}
    & \geq
    \gamma\log k \cdot \Theta^{i,\canonFam} \quad \forall i \in [k],  \label{c2S}\tag{$\overline{\mbox{C2}}$}
  \end{align}
  where $\canonFam$ is the canonical suffix chain family defined for the modified
  instance $\I[\late]$ of \pdls. 
\end{lemma} 
\iffull{
\begin{proof}
	Observe that since $\bar x = \gamma \log k \cdot x$, we can view
	constraint~\eqref{c1S} as being precisely inequality~\eqref{eq:1}
	with both sides scaled up by a factor of $\gamma\log k$; similarly,
	we can also view constraint~\eqref{c2S} as constraint~\eqref{c2p}
	scaled up by this same factor.  Thus, the lemma follows immediately
	from inequality~\eqref{eq:1} and the fact that $(U^0,x^0)$ is feasible
	for (P').
\end{proof}}

We now randomly round $(\bar{U},\bar{x})$ to an integral solution
$(\hat U,\hat x)$ as follows.  For each $P\in\P$, we independently
select a single random suffix chain $\suffChain\in\S^P$ using
marginals derived from $\bar x$, and set the corresponding
$\hat{x}_{\suffChain}=1$.  In particular, we set $\hat{x}$ so that for
all $P\in\P$ and all $\suffChain\in\S^P$ we have
\begin{equation*}
  \Pr[\hat{x}_{\suffChain}=1]
  =
  \begin{cases}
    \bar{x}_{\suffChain}
    &\qquad\text{if $\suffChain\in \T$}\\
    1-\sum_{{\suffChain}'\in\T}\bar{x}_{\suffChain'}
    &\qquad\text{if $\suffChain= \canonFam^P$.}
  \end{cases}
\end{equation*}
Since $(\bar U, \bar x)$ satisfies~\eqref{c1S}, we can see that the
above describes a valid randomized process.  We run this process
independently for each path $P\in\P$ to obtain $\hat x$.  A job $j \in
[n]\setminus \late$ is late if it is contained in level $i(j)$ of the
suffix chain $\suffChain$ chosen for path $P$ by the above process. 
Thus, we set
\[ \hat{U}_j := \sum_{\suffChain\in\S^P:j \in \suffChain_{i(j)}} \hat{x}_{\suffChain}, \]
We now claim
that the expected value of $\hat{U}_j$ is precisely $\bar{U}_j$.

\begin{lemma}
  \label{lem:exp_cost}
  For all $j \not\in \late$, $E[\hat{U}_j]= \bar{U}_j$.
\end{lemma}
\iffull{
  \begin{proof}
    Let $P$ be the path containing $j$, and consider a chain $\suffChain\in\S^P$
    such that $j\in \suffChain_{i(j)}$. The probability for
    $\hat{x}_{\suffChain}$ to be $1$ is precisely $\bar{x}_{\suffChain}$, and hence it immediately follows that
    \[ E[\hat{U}_j] = \sum_{\suffChain\in\S^P: j \in \suffChain_{i(j)}}  \Pr[\hat{x}_{\suffChain}=1]
    = \sum_{\suffChain\in\S^P:j \in \suffChain_{i(j)}}\bar{x}_{\suffChain}=\bar{U}_j. \]
  \end{proof}}

The preceding lemma shows that the expected penalty of $(\hat U,\hat x)$
in the modified instance $\I[\late]$ is exactly $\sum_{j \in [n]\setminus
  \late} w_j\bar{U}_j$. The following lemma shows that the schedule induced by
$\hat{x}$ postpones at least $\Theta^{i,\canonFam}$ jobs past deadline $D_i$ for all $i
\in [k]$ with constant probability. 

\begin{lemma}\label{lem:chernoff}
  With constant probability, we have
  \begin{equation}\label{eq:feas}
    \sum_{P \in \P} \sum_{\suffChain\in\S^P} p^{i,\canonFam}_{\suffChain}
    \hat{x}_{\suffChain} \geq \Theta^{i,\canonFam} \quad \forall i \in [k], 
  \end{equation}
    where $\canonFam$ is the canonical suffix chain family for the modified
    \pdls\ instance $\I[\late]$.  In particular, for $\gamma=4$ the constraint
    holds with probability at least $0.7$.
\end{lemma}
\iffull{
\begin{proof} \emph{(of Lemma \ref{lem:chernoff})}
  Our proof relies on two bounds on random variables.  Before
  proceeding with the proof itself, we begin by stating the two required bounds for the sake of completeness.

  First, we need the following form of Bernstein inequality~\cite{janson2002concentration}.  Let $X_1,\dots,X_n$ be
  independent, nonnegative random variables uniformly bounded by some
  $M\ge0$, i.e.~such that $\Pr[X_i \le M]=1$ for all $i$.
  Then, if we let $X=\sum_iX_i$, we have that
  \begin{equation}
    \label{eq:chernoff}
    \Pr[X \le \E[X] - \lambda]
    \le
    \exp\left(
    -\frac{\lambda^2}{2\left(\Var(X)+\lambda M /3\right)}
    \right),
  \end{equation}
  for any $\lambda\ge0$.

  Second, we need the Bhatia-Davis Inequality~\cite{BD00}, which
  states that for any random variable $X$ with support in the interval
  $[m,M]$, i.e.~such that $\Pr[m \le X \le M]=1$, we have
  \begin{equation}
    \label{eq:bhatia-davis}
    \Var(X) \le (M-\E[X])(\E[X]-m).
  \end{equation}

  We now describe how we apply the above bounds to achieve the desired
  bound on the probability that~\eqref{eq:feas} is satisfied.  Fix some
  $i\in[k]$.  Define random variables $\{X_P\}_{P \in \P}$ as 
  \begin{equation*}
    X_P := \sum_{\suffChain\in\S^P} p^{i,\canonFam}_{\suffChain}\hat{x}_{\suffChain}
  \end{equation*}
  where $\canonFam$ is the canonical chain family defined for the
  modified instance of \pdls; let $X=\sum_{P\in\P}X_P$ denote the sum
  of these random variables.  We make the following observations on
  the random variables $X_P$:
  \begin{itemize}
  \item for $P,P'\in\P$, $P\neq P'$, we have that $X_P$ and $X_{P'}$
    are independent since our rounding process made independent choices for the two paths;
  \item for each $P\in\P$ we know $X_P$ is nonnegative, since
    $p^{i,\canonFam}_{\suffChain}$ and $\hat{x}_{\suffChain}$ are always nonnegative;
  \item for each $P\in\P$ we can see that we always have 
    \begin{equation*}
      X_P 
      =\sum_{\suffChain\in\S^P} p^{i,\canonFam}_{\suffChain}\hat{x}_{\suffChain}
      \le \max_{\suffChain\in\S^P} p^{i,\canonFam}_{\suffChain}
      \le \Theta^{i,\canonFam},
    \end{equation*}
    where the first inequality follows by constraint~\eqref{c1S}, and the second inequality follows by the definition of $p^{i,\canonFam}_{\suffChain}$; and
  \item the expectation of $X=\sum_{P \in \P}X_P$ satisfies
    \begin{equation*}
      \E[X] 
      = \E\Big[\sum_{P \in \P}\sum_{\suffChain\in\S^P} p^{i,\canonFam}_{\suffChain}\hat{x}_{\suffChain}\Big]
      = \sum_{P \in \P}\sum_{\suffChain\in\S^P} p^{i,\canonFam}_{\suffChain}\bar{x}_{\suffChain}
      \ge \gamma \log k\cdot\Theta^{i,\canonFam}.
    \end{equation*}
    The inequality above follows by constraint~\eqref{c2S} in
    Lemma~\ref{lem:boostfeas}; the second equality above follows by
    observing that we always have
    $\Pr[\hat{x}^i_{\suffChain}=1]=\bar{x}^i_{\suffChain}$ in our sum.
  \end{itemize}
  
  With the above observations in hand, we apply
  the Bhatia-Davis Inequality~\eqref{eq:bhatia-davis} to get that
  \begin{align*}
    \Var(X)
    =\sum_{P\in\P}\Var(X_P)
    &\le\sum_{P\in\P}(\Theta^{i,\canonFam}-\E[X_P])(\E[X_P]-0)\\
    &\le\sum_{P\in\P}\Theta^{i,\canonFam}\E[X_P]\\
    &=\Theta^{i,\canonFam}\E[X],
  \end{align*}
  where the first equality follows by the fact that the variables
  $X_P$ are independent.  Thus, applying the  Bernstein inequality~\eqref{eq:chernoff} with $M=\Theta^{i,\canonFam}$ and
  $\lambda=\E[X]-\Theta^{i,\canonFam}$ gives us that
  \begin{align*}
    \Pr[X \le \Theta^{i,\canonFam}]
    &\le\exp\left(-\frac{\left(\E[X]-\Theta^{i,\canonFam}\right)^2}{2\left(\Theta^{i,\canonFam}\E[X]+(\E[X]-\Theta^{i,\canonFam})\Theta^{i,\canonFam}/3\right)}\right)\\
    &\le\exp\left(-\frac{\left(\E[X]-\Theta^{i,\canonFam}\right)^2}{(4/3)\left(\Theta^{i,\canonFam}\E[X]\right)}\right)\\
    &=\exp\left(-\frac{3}{4}\left(\frac{\E[X]}{\Theta^{i,\canonFam}}\right)\left(1-\frac{\Theta^{i,\canonFam}}{\E[X]}\right)^2\right).
  \end{align*}
  The first inequality above follows from the previously mentioned
  application of the Bernstein Inequality, the second by observing
  $\Theta^{i,\canonFam}\ge0$ always and gathering like terms in the
  denominator, and the equality follows by pulling a factor of
  $(\E[X])^2$ out of the numerator.  As noted earlier, however, we
  have that $\E[X]\ge\gamma\log k\cdot\Theta^{i,\canonFam}$; taking
  $\gamma=4$ and substituting this into the above gives us that, in
  fact, $\Pr[X\le\Theta^{i,\canonFam}]\le\frac{3}{10k}$.  Since the
  constraint~\eqref{eq:feas} holds for a given $i\in[k]$ if and only
  if $X\ge\Theta^{i,\canonFam}$, by taking a union bound we can see
  that the constraint holds for all $i\in[k]$ with probability at
  least $0.7$.
\end{proof}}

For each $P \in \P$ let $\suffChainRnd^P$ be the join of the suffix chain
corresponding to solution $\hat x$, and the canonical suffix chain $\canonFam^P$;
i.e., suppose that $\hat{x}_{\suffChain}=1$ for $\suffChain \in \S^P$. Then
\begin{equation} \label{eq:suffix}
  \suffChainRnd^P = \suffChain \vee \canonFam^P. 
\end{equation}
Clearly, $\suffChainRnd^P$ is a suffix chain for path $P$. 
We use the following greedy algorithm to obtain a schedule. 

\begin{leftbar}
\begin{algorithmic}
\FOR{$i=1$ to $k$}
\FORALL{$P \in \P$}
\STATE Schedule all jobs in $P\setminus \suffChainRnd^P_{i}$ not already scheduled respecting the precedence constraints 
\ENDFOR
\ENDFOR
\STATE Schedule all remaining jobs respecting the precedence constraints
\end{algorithmic}
\end{leftbar}

\begin{theorem}
  \label{thm:sched_good}
  The schedule produced by the above algorithm is feasible.
  Furthermore, if \eqref{eq:feas} holds, the schedule has cost at most
  $\sum_{j\notin \late}w_j\hat{U}_j$ in the instance $\I[\late]$.
\end{theorem}
\iffull{\begin{proof} 
  We begin by noting that the schedule produced by the proposed
  algorithm respects the precedence constraints of all $P\in\P$. This
  follows as $\suffChainRnd^P$ is a suffix chain for all $P
  \in \P$, and hence, the algorithm schedules the jobs in
  $\suffChainRnd^P_{i-1}\setminus \suffChainRnd^P_i$ in iteration $i$ for $P \in
  \P$ in precedence order. 

  Next, we show that whenever~\eqref{eq:feas} holds, we have that the
  penalty of the schedule produced by our algorithm is at most
  $\sum_{j\notin \late}w_j\hat{x}_j$ in the instance $\I[\late]$.  Note
  that the total processing time of all jobs scheduled after iteration
  $i$ by our algorithm is 
  \begin{align*}
    \sum_{P\in\P}\sum_{j\in\suffChainRnd^P_i}p_j
    &\ge\sum_{P\in\P}\left( \sum_{\suffChain\in\S^P} p^{i,\canonFam}_{\suffChain}
    \hat{x}_{\suffChain} +\sum_{j\in \canonFam^P_i}p_j\right)\\
    &\ge \Theta^{i,\canonFam} + \sum_{P\in\P}\sum_{j\in \canonFam^P_i}p_j\\
    &\ge \Gamma - D_i,
  \end{align*}
  where $\Gamma$ is the total processing time of all jobs.  The first inequality
  above follows by the definitions of $p^{i,\canonFam}_{\suffChain}$ and $\suffChainRnd$; the
  second by our assumption that~\eqref{eq:feas} holds; and the third
  by the definition of $\Theta^{i,\canonFam}$.  This means, however, that all
  of the jobs scheduled during iterations $1,2,\dots,i$ of our
  algorithm will be completed by time $D_i$.  Now consider a job $j\in
  P\setminus \late$
  with deadline $d_j=D_{i(j)}$. If $j$ is late in the given schedule,
  then it must have been
  scheduled {\em after} iteration $i(j)$. This implies that 
  $j\in\suffChainRnd^P_{i(j)}$, and so
  \[ \hat{U}_j=\sum_{\suffChain\in\S^P:j \in \suffChain_{i(j)}}\hat{x}_{\suffChain}=1.\]
  Thus, if we let
  $\mathcal{L}$ be the set of jobs $j$ that are late in the schedule
  produced by our algorithm, we can see that the penalty of that schedule in the modified instance $\I[\late]$ is 
  \begin{equation*}
    \sum_{j\in\mathcal{L}}w_j \le \sum_{j\notin \late} w_j\hat{U}_j,
  \end{equation*}
  exactly as claimed.
\end{proof}}

\begin{corollary}
  \label{cor:sched_cost}
  The schedule produced by the above algorithm is feasible and incurs
  penalty at most $8 \log k \cdot \sum_j w_j U_j$ in the original
  instance of the \pdls with constant probability.
\end{corollary}
\iffull{\begin{proof}\emph{(Proof of Corollary \ref{cor:sched_cost})}
	Recall how we arrived at the integral solution
	$(\hat{U},\hat{x})$. Given a solution 
	$(U,x)$ for \eqref{p}, we define the set $\late$ of jobs $j$ whose indicator
	variables $U_j$ have value at least $1/(\gamma \log k)$. We then 
	obtained a modified instance $\I[\late]$ of \pdls\ by increasing the deadlines of
	jobs in $\late$ to $\Gamma$. Assuming that $(U,x)$ satisfies constraint 
	\eqref{c2} for the canonical suffix chain family $\canonFam$ for $\I[\late]$ we
	then generated a new solution $(\bar{U},\bar{x})$ such that
	\begin{equation*}
	\sum_{j\notin \late}w_j\bar{U}_j + \sum_{j\in \late}w_j \le \gamma \log k \sum_jw_jU_j.
	\end{equation*}
	We then rounded our solution $(\bar{U},\bar{x})$ to produce the
	integral solution $(\hat{U},\hat{x})$ to the modified instance.
	
	Now, by combining Theorem~\ref{thm:sched_good} and
	Lemma~\ref{lem:chernoff}, we can see that with probability at least $0.7$
	we get a feasible schedule whose cost is at most
	$\sum_{j\notin \late}w_j\hat{U}_j$ in the modified instance, while setting the
	parameter $\gamma=4$.  Consider how the cost of this schedule can
	change between the modified instance and the original instance:
	since our modification was precisely to set the deadline of every
	job $j\in \late$ to $\Gamma$, we know that the cost of our schedule in
	the original setting can be at most
	$
	\sum_{i\notin \late}w_j\hat{U}_j+\sum_{j\in \late}w_j.
	$
	
	Now, by Lemma~\ref{lem:exp_cost}, we have that $\E[\sum_{i\notin
		\late}w_j\hat{U}_j]=\sum_{i\notin\late}w_j\bar{U}_j$.  Thus,
	Markov's Inequality gives us that
	$\sum_{i\notin\late}w_j\hat{U}_j\le2\sum_{i\notin\late}w_j\bar{U}_j$
	with probability at least $1/2$; taking a union bound over this and
	our probability of our rounding procedure producing a feasible
	schedule, we can conclude that with probability at least $0.2$, we
	produce a feasible schedule the cost of which in the original
	instance is at most
	\begin{equation*}
	2\sum_{i\notin \late}w_j\bar{U}_j+\sum_{j\in \late}w_j
	\le 2\cdot 4 \log k \cdot \sum_jw_jU_j,
	\end{equation*}
	exactly as claimed.
\end{proof}}

\section{Solving LP \eqref{p}}
\label{sec:projection}

\ifconf{
In this section we show that \eqref{p} has a  {\em compact}
reformulation with a polynomial number of variables when precedence
constraints are chain-like.}

\iffull{
While we have shown how we can express the $\pdls$ problem as an IP in
Section~\ref{sec:ex_ip}, and how we can round solutions to a weakened
LP relaxation in Section~\ref{sec:ex_ip_rounding}, an important step
in the process remains unspecified: how do we find a solution $(U,x)$
to round?  This is especially problematic as the LP has an exponential
(in $n$) number of both variables {\em and} constraints, and it is not
clear how to solve such LPs in general. In this section, we show that
\eqref{p} has a compact reformulation when precedences are chain-like.}

\iffull{\subsection{An IP formulation with polynomial number of variables}}

Using the specific shape of precedences, we show how the important
suffix-chain structure of postponed jobs can be captured more
compactly. This allows us to reduce the number of variables
drastically while slightly increasing the number of
constraints. Roughly speaking, our new LP {\em decouples} decisions on
job-postponement between the layers in $[k]$.

The new IP has a binary variable $x^i_j$ for every job $j \in [n]$ and
for all deadlines $D_i \in \D$. In a solution $x^i_j=1$ if job $j$ and
all of its successors are executed after deadline $D_i$ and all of job
$j$'s predecessors are executed before deadline $D_i$.  We can see
that our definition of the variables $x_j^i$ ensures the desired
suffix structure on every path $P\in\P$ (see Observation~\ref{obs1});
so we need only add constraints to ensure that the chosen suffixes for
a path $P$ form a chain (in the sense of Observation~\ref{obs2}).  

\para{\eqref{d1} At most one suffix of postponed jobs per path per layer}
We want to choose at most one suffix of jobs to
defer for each path $P\in\P$ on every layer $i\in[k]$. This yields the
following constraint corresponding to~\eqref{c1}:
\begin{equation}\label{d1}\tag{D1}
  \sum_{j \in P} x^i_j \leq 1 \quad \forall P \in \P, i \in [k]. 
\end{equation}

\para{\eqref{d2} Deferring sufficiently many jobs} 
As in our previous IP, our new formulation has a constraint for each
suffix chain family $\suffChainFam \in \S$ and each layer $i \in [k]$. As before,
suppose that we were looking for a schedule that defers the jobs in
$\suffChainFam^P_i$ past deadline $D_i$ for all $P \in \P$, and for all $i \in
[k]$. Then define $\Theta^{i,\suffChainFam}$ as in \eqref{thetais}, and let
\[
p^{i,\suffChainFam}_j := \min \left\{ \sum_{j' : j \preceq j', j' \not\in \suffChainFam^P_i} p_{j'},
\Theta^{i,\suffChainFam} \right\},
\]
for every job $j\in[n]$.  The new constraint corresponding
to~\eqref{c2} is now:
\begin{equation}\label{d2}\tag{D2}
\sum_{P \in \P} \sum_{j \in P\setminus \suffChainFam^P_i} p^{i,\suffChainFam}_j x^i_j \geq
\Theta^{i,\suffChainFam} \quad \forall i \in [k], \, \forall \suffChainFam \in \S,
\end{equation}
which enforces that the total (capped) weight of jobs deferred past
$D_i$, beyond what is deferred by $\suffChainFam$ is sufficiently large.

\para{\eqref{d3} Chain structure of postponed suffixes}
We need additional constraints to ensure our new IP chooses suffixes of jobs to defer that exhibit the chain structure that characterizes feasible schedules.  Consider
$i,i' \in [k]$, with $i<i'$, and let $j \in [n]$ be a job whose
execution is postponed until after time $D_{i'}$. Then as previously
observed, its execution must also be postponed until after deadline $D_i$,
by the ordering of deadlines. We capture this with the following
family of constraints:
\begin{equation}\label{d3}\tag{D3}
\sum_{j' :j'\preceq j} x^{i+1}_{j'} \leq \sum_{j': j' \preceq j}
x^{i}_{j'} \quad \forall P \in \P, \forall j \in P, \forall i \in
[k-1].
\end{equation}

We can now state the entire IP. As with~\eqref{p}, we introduce a binary
variable $U_j$ that takes value $1$ if job $j \in [n]$ is postponed
past its deadline $d_j$. For a given job $j\in[n]$, we can define this
variable in terms of the variables in our new IP as follows.  Let
 $d_j=D_{i(j)}$. Then job $j$ is postponed if it is part of a
chosen suffix for some layer $i \geq i(j)$. We can therefore set
\[ U_j := \max_{i \geq i(j)} \sum_{j' \preceq j} x^i_{j'} = 
\sum_{j' \preceq j} x^{i(j)}_{j'}, \]
where the equality follows from \eqref{d3}. The standard LP relaxation
of the IP is:
\begin{equation} \label{p2} \tag{P2}
  \min \left\{ \sum_{j \in [n]} w_jU_j \,:\, \eqref{d1}, \eqref{d2},
  \eqref{d3}, x \geq 0 \right\}. 
\end{equation}

In the case of chain-like precedences, we are able to obtain an
efficient, objective-value preserving, and invertible map between
fractional solutions to \eqref{p2} and fractional {\em cross-free}
solutions to \eqref{p}.

\begin{theorem}
  \label{thm:project}
  Consider
  $\{x_{\suffChain}\}_{\suffChain\in\S^P,P\in\P}$, with $x \ge 0$, and let 
\begin{equation}
  \label{eqn:project}
  \tilde{x}_j^i=\sum_{{\suffChain\in\S^P:\suffChain_i=\suffix{j}}}x_{\suffChain},
\end{equation}
for all $i\in[k]$, $P\in\P$, and $j\in P$. Then:
$(i)$ $\tilde x$ satisfies  condition~\eqref{d3};
$(ii)$ $\tilde x$ satisfies conditions~\eqref{d1}
if and only if $x$ satisfies conditions~\eqref{c1}; 
$(iii)$ for any $\suffChainFam \in \mathcal S$, 
$\tilde x$ satisfies \eqref{d2} for $\suffChainFam$ if and only if 
$x$ satisfies \eqref{c2} for family $\suffChainFam$.
Furthermore, the objective value of $x$ in \eqref{p} equals that of
$\tilde{x}$ in \eqref{p2}. 
\end{theorem}
\iffull{\begin{proof} 
  We begin by showing that equation~\eqref{eqn:project} implies that
  $\tilde{x}$ satisfies condition~\eqref{d3}. Fix $P\in\P$, $j\in P$,
  and $i\in[k-1]$, and observe that
  \begin{equation*}
    \sum_{j' :j'\preceq j} \tilde{x}^{i+1}_{j'} 
    = 
    \sum_{j' :j'\preceq j} \sum_{\substack{\suffChain\in{\S^P}:\\\suffChain_{i+1}=\suffix{j'}}}x_{\suffChain} 
    =
    \sum_{{\suffChain\in{\S^P}:j\in\suffChain_{i+1}}}x_{\suffChain}.
  \end{equation*}
  If we apply the same transformation to $\sum_{j': j' \preceq j}
  \tilde{x}^{i}_{j'}$, we see that condition~\eqref{d3} is equivalent to
  \begin{equation}\label{eq:d3equiv}
    \sum_{{\suffChain\in{\S^P}:j\in\suffChain_{i+1}}}x_{\suffChain}
    \le
    \sum_{{\suffChain\in{\S^P}:j\in\suffChain_{i}}}x_{\suffChain}.
  \end{equation}
  Since every $\suffChain\in\S^P$ is a chain we have
  $\suffChain_{i+1}\subseteq\suffChain_i$, and thus every summand on the
  left side of the inequality also appears on the right.  Since $x
  \geq \0$, it follows that both inequality~\eqref{eq:d3equiv} and constraint~\eqref{d3} hold.

  To see that (ii) holds, we focus on path $P \in \P$, and $i \in
  [k]$, and observe that
  \begin{equation*}
    \sum_{j \in P} \tilde{x}^i_j 
    = \sum_{j \in P}\sum_{\substack{\suffChain\in\S^P:\\\suffChain_i=\suffix{j}}}x_{\suffChain}
    = \sum_{\suffChain \in \S^P} x_{\suffChain}.
  \end{equation*}
  It immediately follows that $\tilde{x}$ satisfies \eqref{d1} iff $x$
  satisfies \eqref{c1}. 
  
  Next, we show that constraint (C2) holds for $\tilde{x}$ if and only
  if constraint (D2) holds for $x$. Consider a chain family $\suffChainFam$ on $\P$,
  and a layer $i\in[k]$. For $P\in\P$ and $j\in
  P$, we can see that for any $\suffChain\in\S^P$ such that $\suffChain_i=\suffix{j}$ we have 
  \begin{equation*}
    p^{i,\suffChainFam}_j 
    = \min \left\{ \sum_{j' : j \preceq j', j' \not\in \suffChainFam^P_i} p_{j'},\Theta^{i,\suffChainFam}\right\}
    =p^{i,\suffChainFam}_{\suffChain}.
  \end{equation*}
  If we apply~\eqref{eqn:project} followed by the above equality, we get
  that
  \begin{align*}
 \sum_{P \in \P} \sum_{j \in P\setminus \suffChainFam^P_i} p^{i,\suffChainFam}_j \tilde{x}^i_j 
 &=\sum_{P \in \P} \sum_{j \in P\setminus \suffChainFam^P_i} p^{i,\suffChainFam}_j\sum_{\substack{\suffChain\in\S^P:\\\suffChain_i=\suffix{j}}}x_{\suffChain}  \\
 &=\sum_{P \in \P} \sum_{j \in P\setminus \suffChainFam^P_i}\sum_{\substack{\suffChain\in\S^P:\\\suffChain_i=\suffix{j}}} p^{i,\suffChainFam}_{\suffChain}x_{\suffChain}  \\
 &=\sum_{P \in \P} \sum_{\suffChain\in\S^P} p^{i,\suffChainFam}_{\suffChain}x_{\suffChain},
  \end{align*}
  where the final equality follows by combining the inner two
  summations and observing that $p_{\suffChain}^{i,\suffChainFam}=0$ whenever
  $\suffChainFam_i^P\supseteq\suffChain_i$.  Thus we have that $\tilde{x}$ satisfies
  the constraints in~\eqref{d2} corresponding to $\suffChainFam$ if and only if
  $x$ satisfies the constraints in~\eqref{c2} corresponding to $\suffChainFam$.
  
  Finally, fix some $P\in\P$.  Using \eqref{eqn:project} and the definition of
  $\tilde{U}_j$, we see that
  \begin{equation*}
    \sum_{j\in P}w_j\tilde{U}_j
    =\sum_{j\in P}\sum_{j' \preceq j} w_j\tilde{x}^{i(j)}_{j'}
    =\sum_{j\in P}\sum_{j' \preceq j}\sum_{\substack{\suffChain\in\S^P:\\\suffChain_{i(j)}=\suffix{j'}}}w_jx_{\suffChain}
    =\sum_{j\in P}\sum_{\substack{\suffChain\in\S^P:\\j\in\suffChain_{i(j)}}}w_jx_{\suffChain},
  \end{equation*}
  where the final equality follows simply by combining the two inner
  summations.  Note, however, that if we change the order of the two
  summations on the right-hand side we obtain
  \begin{equation*}
    \sum_{j \in P}w_j\tilde{U}_j
    =\sum_{\suffChain\in\S^P}\sum_{\substack{j\in P:\\j\in\suffChain_{i(j)}}}w_jx_{\suffChain}
    =\sum_{\suffChain\in\S^P}w_{\suffChain}x_{\suffChain}.
  \end{equation*}
  Summing the above over $P\in\P$, it follows that $\tilde{x}$ and $x$ have the same objective
  values in~\eqref{p2} and~\eqref{p}, respectively. 
\end{proof}}

The above theorem \ifconf{(see \cite{full} for its proof)} immediately implies a natural algorithm for
efficiently constructing a solution $\tilde{x}$ to~\eqref{p2} from a
given solution $x$ to~\eqref{p}. It fails to show how to perform the
inverse operation, however. We now provide the missing piece \ifconf{(once
more, see \cite{full} for its proof)}.

\begin{theorem}
  \label{thm:lift}
  Given a solution $\tilde{x}$ to \eqref{p2}, we can efficiently
  construct a cross-free solution $x$ to \eqref{p} that  satisfies condition~\eqref{eqn:project}
  from Theorem~\ref{thm:project}.
\end{theorem}
\iffull{\begin{proof} 
  We begin by constructing the collection of suffix chains that lie in
  the support of our claimed solution $x$ for~\eqref{p}; we then
  describe the values $x$ associates with each of the suffix chains in
  this collection; and finally we argue that $x$ satisfies the claimed
  properties.

  Consider a path $P\in\P$, $i \in [k]$ and $\alpha \in [0,1]$. Let
  \begin{equation}
    \label{eqn:barS_alpha}
    \suffChain_i(\alpha)
    :=
    \Bigg\{ 
    j \in P : 
    \sum_{j'\preceq j}\tilde{x}_{j'}^i \ge \alpha 
    \Bigg\}, 
  \end{equation}
  and note that the non-negativity of $\tilde{x}$ implies that
  $\suffChain_i(\alpha)$ is a suffix of $P$.
  Furthermore, constraint \eqref{d3} implies that 
  \[
  \suffChain_1(\alpha)\supseteq\suffChain_2(\alpha)\supseteq\dots\supseteq\suffChain_k(\alpha),\] 
  and $\suffChain(\alpha)$ is therefore a valid suffix chain for $P$.  
  Furthermore, for $0\le\alpha\le\alpha'\le1$, and $i \in [k]$ one
  easily sees that $\suffChain_i(\alpha) \supseteq \suffChain_i(\alpha')$,
  and hence $\suffChain(\alpha)$ and $\suffChain(\alpha')$ do not cross. 
  With this we now easily bound the number of distinct members of the 
  family $\{\suffChain(\alpha)\}_{\alpha}$. Consider increasing $\alpha$
  continuously from $0$ to $1$, and count the number of times changes
  occur. For each $i\in[k]$, we know
  that $\suffChain_i(\alpha)$ can only take on $\abs{P}+1$ different
  values. Since
  $\alpha\le\alpha'$ implies $\suffChain(\alpha)\preceq\suffChain(\alpha')$,
  it follows that $\suffChain_i(\alpha)$ becomes smaller as
  $\alpha$ increases.  Thus, as we increase $\alpha$ from $0$ to $1$,
  $\suffChain_i(\alpha)$ can only change at most $\abs{P}$ times.  Since
  any time $\suffChain(\alpha)$ changes, at least one $\suffChain_i(\alpha)$
  must change, we may conclude that $\suffChain(\alpha)$ takes on at most
  $k\abs{P}+1$ distinct values.

  Now we show how to construct a  solution $x$
  for~\eqref{p} with support $\{\suffChain(\alpha)\}_{\alpha}$.
  For suffix chain $\suffChain\in\S^P$ we let
  \begin{equation}
    \label{eqn:lift_probs}
    x_{\suffChain}
    :=
    \sup\{\alpha\in[0,1]:\suffChain(\alpha)=\suffChain\}
    -
    \inf\{\alpha\in[0,1]:\suffChain(\alpha)=\suffChain\};
  \end{equation}
  if $\suffChain=\suffChain(\alpha)$ for some $\alpha \in [0,1]$, and we let 
  $x_{\suffChain}:=0$ otherwise.
  
  Finally, we show that~\eqref{eqn:project} holds, i.e.~that
  \begin{equation*}
    \tilde{x}_j^i=\sum_{{\suffChain\in\S^P:\suffChain_i=\suffix{j}}}x_{\suffChain}.
  \end{equation*}
  Fix $i\in[k]$ and $j\in P$, and consider the sum
  on the right hand side of the above equality.  Recall that we
  previously observed that $\alpha\le\alpha'$ implies
  $\suffChain(\alpha)\preceq\suffChain(\alpha')$; combining this with
  equation~\eqref{eqn:lift_probs}, we can easily see that the sum we
  care about will, in fact, be a telescoping sum that simplifies to
  \begin{equation*}
    \sum_{{\suffChain\in\S^P:\suffChain_i=\suffix{j}}}x_{\suffChain}
    =
    \sup\{\alpha:\suffChain_i(\alpha)=\suffix{j}\}
    -
    \inf\{\alpha:\suffChain_i(\alpha)=\suffix{j}\}.
  \end{equation*}
  By our definition of $\suffChain_i(\alpha)$, however, we can see that
  $j\in\suffChain_i(\alpha)$ if and only if $\sum_{j'\preceq j}x_{j'}^i
  \ge \alpha$.  Further, $\suffChain_i(\alpha)$ can only include elements
  strictly preceding $j$ in $P$ if $\sum_{j'\prec j}x_{j'}^i \ge
  \alpha$.  Thus, we may conclude that 
  \begin{equation*}
    \sup\{\alpha:\suffChain_i(\alpha)=\suffix{j}\}
    -
    \inf\{\alpha:\suffChain_i(\alpha)=\suffix{j}\}
    =
    \sum_{j'\preceq j}\tilde{x}_{j'}^i
    -
    \sum_{j'\prec j}\tilde{x}_{j'}^i
    =
    \tilde{x}_j^i,
  \end{equation*}
  exactly as required.
\end{proof}}

\ifconf{ LP \eqref{p2} has a polynomial number of variables, but an
  exponential number of constraints. A natural strategy to solve such
  an LP would involve the {\em Ellipsoid method}
  \cite{GLS81} that allows us to reduce the problem of
  solving \eqref{p2} to that of efficiently {\em separating} an
  infeasible point from \eqref{p2}. For a candidate solution
  $(U,x)$ to \eqref{p2}, we need to efficiently decide 
  whether it is feasible, and if not, return a violated inequality. We
  do not know how to solve the separation problem for \eqref{p2}, and
  it is in fact not known how to separate KC inequalities efficiently
  in general.
 
  We will overcome this issue following the methodology proposed in
  \cite{CF+00} and relying on a {\em relaxed separation} oracle.  For
  this, we consider a relaxation (P2') of \eqref{p2} where we include
  the Knapsack cover constraint \eqref{d2} only for a certain
  polynomial-sized collection \S' of suffix chain families. Theorems
  \ref{thm:project} and \ref{thm:lift} apply as well if \eqref{p2} is
  replaced by (P2'). Furthermore, we can show that family \S' can be
  chosen such that the solution $x$ for \eqref{p} satisfies conditions
  (a) and (b) from Section \ref{sec:ex_ip_rounding}, and hence can be
  rounded into a low-penalty feasible schedule with constant
  probability. Details are provided in \cite{full}.
}

\iffull{
\subsection{Solving the relaxation}
\label{subsec:ellipsoid}

The usual way to solve linear programs with an exponential number of
constraints is to employ the {\em ellipsoid} method~\cite{GLS81}. 
The method famously allows us to
reduce the problem of solving \eqref{p2} to that of efficiently
{\em separating} an infeasible point from \eqref{p2}.  For a given
candidate solution $(U,x)$ to \eqref{p2}, it suffices to decide (in
polynomial time) whether it is feasible, and if not, return a violated
inequality. We do not know how to solve the separation problem for
\eqref{p2}, and it is in fact not known how to separate KC
inequalities efficiently in general. 

We will overcome this issue following the methodology proposed in
\cite{CF+00} and relying on a {\em relaxed separation} oracle.  For
this, we consider a relaxation (P2') of \eqref{p2} where we replace
constraints \eqref{d2} by
\begin{equation}\label{d2p}\tag{D2'}
\sum_{P \in \P} \sum_{j \in P\setminus \suffChainFam^P_i} p^{i,\suffChainFam}_j x^i_j \geq
\Theta^{i,\suffChainFam} \quad \forall i \in [k], \, \forall \suffChainFam \in \S'.
\end{equation}
for a subset $\S' \subseteq \S$ that initially contains only
the {\em canonical suffix chain family} $\canonFam$ for 
instance $\I$.

In an iteration, we apply the ellipsoid method to (P2'), and this
generates a point $(\tilde U, \tilde x)$ that is optimal and feasible
for (P2'), but not necessarily feasible for \eqref{p2}. Using the
procedure developed in Section \ref{sec:ex_ip_rounding} we now
attempt to map $(\tilde U,\tilde x)$ to a cross-free solution to
\eqref{p}. Specifically, for parameter $\gamma$ chosen there, we
let  
\[ \late = \{ j \,:\, \tilde U_j \geq 1/(\gamma \log k)\}, \] 
be the collection of jobs whose indicator variable is large; note that
$U$ remains constant under the correspondence between solutions
to~\eqref{p} and~\eqref{p2} as given in Theorem~\ref{thm:project}, and
so this is precisely the set $\late$ used in Section
\ref{sec:ex_ip_rounding}.  As before, we imagine increasing the
deadline for the jobs in $\late$ to $\Gamma$ to produce a modified
instance $\I[\late]$ of the \pdls.

We now check whether $(\tilde U, \tilde x)$ violates \eqref{d2} for
the canonical suffix chain family \canonFam\ of instance
$\I[\late]$. In this case, we add \canonFam\ to $\S'$, and
recurse. Otherwise, we know that we can apply the lifting operation of
Theorem~\ref{thm:lift} to obtain a new candidate solution $(U,x)$
for~\eqref{p} whose support is cross-free.  By
Theorem~\ref{thm:project} we can see that, while $(U,x)$ may not be
feasible for~\eqref{p}, it {\em is} feasible for the relaxation~(P')
used in our rounding procedure, and further has objective value no
larger than the optimum of \eqref{p}.  Thus, applying the rounding procedure
of Section~\ref{sec:ex_ip_rounding} yields a solution to our instance
of $\pdls$ with penalty $O(\gamma\log k)$ times the optimal with constant
probability.

In other words, the process described above consists of applying the
ellipsoid method for solving \eqref{p2} with a separation oracle that
might fail in providing a violated inequality: we stop the algorithm
at the first moment that our separation oracle fails, so as to
guarantee that the number of iterations (and therefore the size of
$\S'$) is anyway polynomially bounded in the number of
variables. The solution output at the end might be infeasible for
\eqref{p2}, but (as discussed above) provides a lower bound on its optimal
value and is feasible for the relaxation (P2'); this ensures we can lift it into a solution for~\eqref{p} that has cross-free support and is feasible
for the relaxation (P'), which are precisely the  properties required by our
rounding procedure.

}

\bibliography{dl_sched}
\bibliographystyle{plain}

\iffull{
  \appendix
  \section{Constant number of paths}
  \label{sec:constant}
  \newcommand{\suffVec}{V}
\newcommand{\nextJob}[1]{\min({#1})}
\newcommand{\rec}[1]{\operatorname{OPT}(#1)}
\newcommand{\cost}{\operatorname{cost}}
\newcommand{\suffVecNext}[1][]{{\suffVec}\setminus{#1}}

In this section, we present an algorithm for solving the
\pdls\ problem when chain-like precedence constraints can be modeled
by a constant number of paths, i.e, $\abs{\P}$ is a constant.  For
each path $P\in\P$, let $\suffVec^P$ be a suffix of path $P$, and let
$\suffVec$ be the vector formed by these suffixes.  Further, for a
given path $P$ and any suffix $\suffVec^P\neq\emptyset$, let $\nextJob{\suffVec^P}$ be the first job in suffix $\suffVec^P$ according to the
precedence order $\preceq$ on jobs, i.e.~ the job we would have to
schedule first among those in $\suffVec^P$.

Now, we are ready to design a dynamic program for solving our
problem. We define $\rec{\suffVec, t}$ to be the smallest overall
postponement cost we can incur on jobs in $\suffVec$, when we schedule
them beginning at time $t$ while respecting precedence
constraints. 
In order to compute the value of $\rec{\suffVec,t}$, we need to
consider how we can schedule the jobs in $\suffVec$.  In particular,
consider the {\em first} job we choose to schedule.  While we can
choose this first job from any path $P$ for which
$\suffVec^P\neq\emptyset$, the precedence constraints enforce that it
must always be the earliest job in $\suffVec^P$ i.e.~it must be
$\nextJob{\suffVec^P}$. Define $\cost(j,t)$ as 
\begin{equation*}
  \cost(j,t)
  =
  \begin{cases}
    w_j&\qquad\text{if }d_j<t+p_j\text{; and}\\
    0&\qquad\text{otherwise,}
  \end{cases}
\end{equation*}
to capture the cost of scheduling job $j$ at time $t$.  Then, we
can express $\rec{\suffVec,t}$ recursively as
\begin{equation}
\label{eq:dynamic}
\rec{\suffVec,t}
=\min_{P\in\P:\suffVec^P\neq\emptyset}\left(\cost(j,t) + \rec{\suffVecNext[j],t+p_j}\right),
\end{equation} 
where $j=\nextJob{\suffVec^P}$ and we use $\suffVecNext[j]$ to denote
the vector $\suffVec$ excluding job $j$, i.e.~we have
\begin{equation*}
  (\suffVecNext[j])^P
  =
  \begin{cases}
    \suffVec^P\setminus\{j\}&\qquad\text{if }j\in P\text{; and}\\
    \suffVec^P&\qquad\text{otherwise.}
  \end{cases}
\end{equation*}
Thus, if we abuse notation slightly and let $\P$ denote the vector of
suffix chains that includes every job on every path, and $\emptyset$
denote the vector taking the empty suffix on every path, we can see
that taking base cases of $\rec{\emptyset,t}=0$ for all $t$ and
computing $\rec{\P,0}$ yields precisely the quantity we want to
compute.

{\noindent \em Running time:} While it might appear that the above
recursion describes a dynamic program that runs in pseudo-polynomial
time, due to the second parameter, we note that the first parameter
always fully determines the second.  In particular, for any recursive
call $\rec{\suffVec,t}$ made while computing $\rec{\P,0}$, a simple
induction shows that we always have that
\begin{equation*}
  t = \sum_{P\in\P}\sum_{j\in P\setminus\suffVec^P}p_j.
\end{equation*}
Thus, the number of values we need to compute is bounding by the
number of possible vectors of suffix chains, which is
$\prod_{P\in\P}(\abs{P}+1)=O(n^{\abs{\P}})$.  Since our recurrence takes
the best among at most $\abs{\P}$ possibilities, we can immediately
conclude that we can compute every value of $\rec{\suffVec,t}$ of
interest in time $O(\abs{\P}n^{\abs{\P}})$.

  \section{Single Deadline}
  \label{sec:single_deadline}
  \newcommand{\nexttree}{\operatorname{next}}

In this section, we present a pseudo-polytime algorithm for solving
the pDLS problem in the case where all jobs face a single, common
deadline.  The algorithm we give addresses not just chain-like
precedence constraints, but the more general case of precedence
constraints forming a tree where jobs can have multiple predecessors
(but still no more than one successor).  This matches the setting of
Ibarra and Kim~\cite{IK78}, and highlights the critical role of
multiple, distinct deadlines in the pDLS problem: Lenstra and Rinnoy
Kan\cite{LR80} proved that the pDLS problem is strongly NP-hard even when
all jobs have unit processing time and deferral cost, and precedence
constraints are chain-like; with only a single, common deadline,
however, the existence of a psuedo-polytime algorithm implies that the
problem becomes only weakly NP-Hard, even in the more general case
where precedence constraints form trees.

We solve our problem by designing an appropriate dynamic program.
Before proceeding with the details, we begin by giving some intuition
and making necessary definitions.  Limiting all jobs to share a common
deadline $D$ forces the problem to closely resemble the knapsack
problem, and in fact our dynamic program follows in precisely that
vein.  With only a single deadline, our problem becomes one of
deciding which jobs we choose to include before that deadline, and
which we defer until after the deadline; in other words, we have a
certain amount of time before the deadline occurs, and we need to
decide what jobs we want to use that time on to minimize the cost of
deferred jobs (or, equivalently, maximize the cost of on-time jobs).
In the setting we consider, each job may have multiple predecessors
but only a single sucessor.  This implies that, if we consider any
tree present in the precedence constraints, if we schedule a job $j$
from that tree we must also schedule every job in the subtree rooted
at job $j$.  Thus, our task is to find a subforest of the forest of
precedence constraints.  If we focus on one tree, we can consider
whether or not to schedule the job at the root: if we choose to
schedule the root job, then we must schedule the entire tree; if we
choose not to schedule the root job, then we may decide independently
how (and if) to schedule jobs from the subtrees rooted at each of its
children.  In effect, we choose to either schedule the entire tree,
and then recurse on the remaining trees, or replace the tree with one
tree for each of the root's children, and then recurse on the modified
forest.

In order to make the above intuition concrete, we need the following
definitions.  Let the collection $[n]$ of jobs be indexed according to
a pre-order traversal of the precedence constraint forest, and for any
job $j$, let $T(j)$ denote the subtree rooted $j$, i.e.~all of $j$'s
predecessors.  Define $\nexttree(j)=\min \{j+1,j+1,\dots,n+1\}\setminus
T(j)$.  Note that this implies that every successor of $j$ will have a
strictly earlier index $j'<j$, and the predecessors of $j$ will be precisely those jobs with indexes in the range $\{j+1,j+2,\dots,\nexttree(j)-1\}$/.
Further, let
\begin{align*}
  W(j) &= \sum_{j'\in T(j)}w_{j'}\text{; and}\\
  P(j) &= \sum_{j'\in T(j)}p_{j'}
\end{align*}
denote the total penalty and total processing of all jobs in $T(j)$.
Recall that every job $j$ shares a single common deadline $d_j=D$.

We now describe our dynamic program.  Let $\rec{j,t}$ represent the
minimum penalty we can incur if we only have $t$ processing time in
which to schedule jobs $\{j, j+1,\dots,\}$.  In order to compute this
value, we need to decide whether or not to schedule job $j$ in the
time we have left before the deadline.  If we do schedule $j$, then we
must schedule all of the jobs in $T(j)$.  This uses up processing time
$P(j)$ but incurs no penalty, and we may then make independent
decisions about all remaining jobs.  If we do not schedule job $j$,
then we incur a penalty of $w_j$ but use no processing time, and then
may make independent decisions about all remaining jobs.  Thus, we may
express $\rec{j,t}$ via the recurrence
\begin{align*}
  \rec{j,t}
  &=
  \min\left\{\rec{\nexttree(j), t-P(j)}, w_j + \rec{j+1,t}\right\},
\intertext{with base cases of}
  \rec{j,t}
  &=
  \begin{cases}
    +\infty&\qquad\text{if }t < 0\text{; and}\\
    0 &\qquad\text{if }t\ge0\text{ and }j=n+1.
  \end{cases}
\end{align*}

From the above discussion, we can readily see that the value of
$\rec{1,D}$ will be precisely the minimum penalty that can be occured
when scheduling all of the jobs in $[n]$.

{\noindent \em Running time:} Since we will only need to compute
$\rec{j,t}$ for parameter settings $j\in[n+1]$ and $t\in[D]$, we can
see that our dynamic program will require at most $O(n \cdot D)$
distinct values to be calculated.  Since each calculation is simply a
minimum of at most two options, we can compute the optimal value
$rec{1,D}$ in time $O(n\cdot D)$.

  \section{Correlated rounding}
  \label{sec:correlated}
  In this section, we consider whether we can replace our {\em
  independent} rounding scheme with a {\em dependent} one to improve
our approximation ratio.  We follow an approach similar to that of
Carr et al.~\cite{CF+00} for the knapsack problem.  At a high level,
the key idea in the rounding scheme of~Carr et al.~is to try and
ensure that the possible solutions it might produce are as uniform in
size as possible.  In particular, they do this by ensuring that sets
of items in each potential knapsack solution have size profiles that
are as similar as possible.  In our setting, this would correspond
to trying to ensure that for every deadline $D_i$, the family of
suffixes we defer on the set $\P$ of paths has size that is as
uniform as possible; unfortunately, we are much more constrained
when selecting suffixes than we would be when choosing items for a
knapsack.

Specifically, there are two key difficulties.  First, since each
solution can only use a single suffix from each path, and different
paths can randomize over paths with very different size profiles, it
may be impossible to ensure that every solution has a similar size
profile.  For example, if every path but one defers a negligible
number of jobs past some deadline $D_i$, then we cannot make our
solutions any more uniform in size than the distribution on that single critical path.
Second, even if we can ensure that solutions defer sufficiently
uniform total size of jobs past a given deadline $D_i$, we need a way to
do this for all deadlines in $\D$ {\em simultaneously}; that is, we
need a way to ensure our approach makes consistent choices for every
deadline $D_i$.  This is problematic: even though we
work with suffix chains and so know that, on average, the length of the particular suffix
deferred on a given path decreases with the deadline $D_i$, the {\em rate} of
this decrease could differ greatly between paths. For example, one path $P$ might
defer much larger suffixes (on average) past some deadline $D_i$ than
every other path in $\P$; but if the size of suffixes deferred on $P$ past
deadline $D_{i+1}$ becomes much smaller, while that of suffixes deferred on other paths remains relatively
constant, the situation  becomes exactly the reverse for deadline $D_{i+1}$.  Since the
relative sizes of suffixes deferred on each path can change quite
dramatically between one deadline and the next, it becomes difficult
to devise a scheme that ensures consistent choices for all of the deadlines in $\D$.

It turns out that the first concern above is not difficult to deal with; it
is the second concern that causes difficulties.  When we only have two paths, however, we can overcome the second
concern as well.  In this section, we describe a correlated rounding scheme
for two paths that provides a $2$-approximation for the case of two
paths.  While we could achieve the same factor for two paths using a
naive approach, we are hopeful the technique we describe here can be
extended to more paths without the approximation factor increasing
linearly with the number of paths (as the naive approach's factor
would).

We now describe our correlated rounding scheme for the case where we only
have two paths, say $\P=\{\pOne,\pTwo\}$.  We follow the general outline as the rounding procedure of Section~\ref{sec:ex_ip_rounding}, with two major
changes.  First, we adjust the definition of the set $\late$,
replacing the filtering parameter $\gamma\log k$ by $2$.  Second, we
modify our rounding procedure to be correlated, rather than
independent, for the two chains $\pOne$ and $\pTwo$.  We discuss the two
changes in detail below.

Our first change in the rounding procedure is to adjust the filtering
parameter that splits jobs based on whether their $x_j$ values are
large or small.  In particular, we now consider $x_j$ to be large only
if it is at least $1/2$, rather than when it exceeds $1/\gamma\log k$
as in Section~\ref{sec:ex_ip_rounding}.  This means the set $\late$ of late jobs
defining our modified instance $\I[\late]$ is now
$
  \late = \{j : U_j \ge 1/2\}.
$

Our second change will be to adjust the distribution used in our
randomized rounding procedure.  Before, we rounded the variables for
each path $P\in\P$ independently; now, however, while our rounding
scheme will continue to induce the same marginal probability
distribution on each path, we modify our rounding process so that the
random choices we make on the two paths are strongly dependent on each
other.  Retaining the same marginal distributions ensures the analysis
of Section~\ref{sec:ex_ip_rounding} remains valid up until
Lemma~\ref{lem:chernoff}, with only minor changes to accommodate the
new definition of late jobs $\late$.  The major changes in our
approach occur from that point on: we can replace the concentration
result underlying that lemma with an averaging argument that
leverages the dependence structure we have introduced into our rounding
scheme (see Lemma~\ref{lem:chernoff_sub_corr}).  This new argument allows us to prove a much stronger
 approximation guarantee than we obtained for
independent rounding, which furthermore holds with certainty.  Before detailing how we obtain our new
guarantee, we will briefly review the initial steps of our rounding
procedure.  Since these initial steps only require that our new
rounding procedure induces the same marginals on each path, we defer
further details of the rounding scheme for now.

We begin by briefly recalling the overall structure of the rounding
procedure from Section~\ref{sec:ex_ip_rounding}, suitably modified for our new
definition of the set $\late$ of late jobs (see that section for full
details).  We start with an instance $\I$ of the \pdls, and a solution
$(U,x)$ for our linear program~\eqref{p}.  Then, we define a modified
instance $\I[\late]$ of \pdls in which the set $\late=\{j:U_j\ge1/2\}$
of late jobs all have their deadlines changed to be $\Gamma$ (so they
cannot be late in {\em any} schedule).  Finally, we focus on a
relaxation~(P') of~\eqref{p} in which the set of
knapsack constraints~\eqref{c2} is reduced to just the ones
corresponding to the canonical suffix chain family for the modified
instance $\I[\late]$.

We now describe the process we use to round our initial solution
$(U,x)$.  This process can be applied as long as $(U,x)$ both is feasible
for the relaxed program~(P') and has support
$\{\suffChain\in\S^P:x_{\suffChain}>0\}$ on each path $P\in\P$ that is cross-free;
recall that in Section~\ref{sec:projection}  we saw how to produce
exactly such a solution.  Given a solution $(U,x)$ with these properties, we
first modify it to produce a new solution $(\bar{U},\bar{x})$ where
each job $j\in\late$ is no longer late, and each job $j\notin\late$ is
late to twice the extent it is in $(U,x)$.  Formally, this means we
set $\bar{U}_j=2\cdot U_j$ if $j\notin\late$ and $\bar{U}_j=0$
otherwise.  We correspondingly set
\begin{equation*}
  \bar{x}_{\suffChain}=
  \begin{cases}
    2\cdot x_{\suffChain}&\qquad\text{if $\suffChain$ makes some job $j\notin\late$ late; and}\\
    0&\qquad\text{otherwise.}
  \end{cases}
\end{equation*}
Finally, we randomly round the fractional solution $(\bar{U},\bar{x})$
to produce an integral solution $(\hat{U},\hat{x})$, in  such a way that
$\Pr[\hat{x}_{\suffChain=1}]=\bar{x}_{\suffChain}$ for every path $P\in\P$ and every suffix
chain $\suffChain\in\S^P$.

The same arguments as presented in Section~\ref{sec:ex_ip_rounding} -- with only minor
adjustments -- give us the following critical properties for
$(\bar{U},\bar{x})$ and $(\hat{U},\hat{x})$.  The two lemmas below are
direct analogs of Lemmas~\ref{lem:boostfeas} and~\ref{lem:exp_cost}, respectively.  We state
them without proof, as they follow from the same arguments as those
given for their counterparts in Section~\ref{sec:ex_ip_rounding}; we refer the reader to
that section for details.
\begin{lemma}
  \label{lem:boostfeas_corr}
  $(\bar{U}, \bar{x})$ satisfies
  \begin{align}
    \sum_{\suffChain\in \S^P} \bar{x}_{\suffChain} & \leq
    1 \tag{$\overline{\mbox{C1}}$} \quad \forall i \in [k], \forall P
    \in \P\\
    \sum_{P \in \P} \sum_{\suffChain\in\S^P}  p^{i,\canonFam}_{\suffChain}\bar{x}_{\suffChain}
    & \geq
    2\cdot \Theta^{i,\canonFam} \quad \forall i \in [k],  \tag{$\overline{\mbox{C2}}$}
  \end{align}
  where $\canonFam$ is the canonical chain family defined for the modified
  instance $\I[\late]$ of \pdls.
\end{lemma}
\begin{lemma}
  \label{lem:exp_cost_corr}
  For all $j \not\in \late$, $E[\hat{U}_j]= \bar{U}_j$.
\end{lemma}

With the above two lemmas in hand, we are now ready to describe the
rounding scheme we use to obtain $(\hat{U},\hat{x})$ in detail, and
prove our approximation guarantee.

We round our fractional solution $(\bar{U},\bar{x})$ to the integral
solution $(\hat{U},\hat{x})$ as follows.  For a given $P\in\P$, for
every $\alpha\in[0,1]$ we define the chain
\begin{equation*}
  \suffChain^P(\alpha):=\min\{\suffChain\in\supp(\bar{x})\cap\S^P:{\textstyle \sum_{\suffChain'\preceq \suffChain}}\bar{x}_{\suffChain'}>\alpha\},
\end{equation*}
where
$\supp(\bar{x})=\cup_{P\in\P}\{\suffChain\in\S^P:\bar{x}_{\suffChain}>0\}$.
The minimum in the above definition is with respect to the partial
order $\preceq$ on suffix chains; recall that since we know
$\supp(\bar{x})$ is cross-free, this is well-defined.  If the set in
the above definition is empty, i.e.~we have that
$\sum_{\suffChain\in\S^P}\bar{x}_{\suffChain}\le\alpha$, then we
define $\suffChain^P(\alpha)=\canonFam^P$, where $\canonFam$ is the
canonical chain family for the modified instance $\I[\late]$.  We now
round $\bar{x}$ to $\hat{x}$ as follows.  Draw a single uniform random
variable $\alpha\sim\mathcal{U}[0,1]$, and index the pair of paths in
our instance as $\P=\{\pOne,\pTwo\}$.  Our rounding procedure needs to
choose one suffix chain for each of the paths; we use $\alpha$ to
correlate our choices in the following manner.  For paths $\pOne$ and
$\pTwo$ we select the suffix chains
$\suffChain^1=\suffChain^{\pOne}(\alpha)$ and
$\suffChain^2=\suffChain^{\pTwo}(1-\alpha)$, respectively.  At a high
level, our goal is to use $\alpha$ to correlate our choices on the two
paths, pairing large suffix chains on one path with small suffix
chains on the other, and vice versa.  By balancing our choices on the
two paths in this way, we ensure that the {\em combined} weight of the
suffix chains we choose is always relatively large.  We make this
idea concrete by first defining the rounded solution
$(\hat{U},\hat{x})$ and then formalizing the above observation as a
lemma.

As stated above, we want our rounded solution to schedule jobs on
paths $\pOne$ and $\pTwo$ according to the suffix chains
$\suffChain^1=\suffChain^{\pOne}(\alpha)$ and
$\suffChain^2=\suffChain^{\pTwo}(1-\alpha)$, respectively, where
$\alpha\sim\mathcal{U}[0,1]$.  Thus, we set
$\hat{x}_{\suffChain^1}=\hat{x}_{\suffChain^2}=1$,  and set $\hat{x}_{\suffChain}=0$
for all other $\suffChain\in\S^{\pOne}\cup\S^{\pTwo}$.  Correspondingly, for
each job $j$ such that $j\notin\late$, we set $\hat{U}_j=1$ if
$j\in\suffChain^1_{i(j)}$ or $j\in\suffChain^2_{i(j)}$, respectively,
depending on whether $j\in \pOne$ or $j\in \pTwo$; for all other $j\in[n]$
we set $\hat{U}_j=0$.  The following lemma shows that the rounding
scheme outlined above produces the same marginal probabilities for
$(\hat{U},\hat{x})$ as those in Section~\ref{sec:ex_ip_rounding}, thereby establishing
the validity of Lemmas~\ref{lem:boostfeas_corr} and~\ref{lem:exp_cost_corr}.

\begin{lemma}
  \label{lem:corr_marginals}
  When $(\bar{U},\bar{x})$ is rounded to $(\hat{U},\hat{x})$ as
  described above, we have that for all $P\in\P$ and all $\suffChain\in\S^P$,
  \begin{equation*}
    \Pr[\hat{x}_{\suffChain}=1]
    =
    \begin{cases}
      \bar{x}_{\suffChain}&\qquad\text{if }\suffChain\neq \canonFam^P\text{; and}\\
      1-\sum_{\suffChain'\in\S^P\setminus\{\canonFam^P\}}\bar{x}_{\suffChain'}&\qquad\text{if }\suffChain=\canonFam^P\text{,}
    \end{cases}
  \end{equation*}
  where $\canonFam$ is the canonical suffix chain family for $\I[\late]$.
\end{lemma}
\begin{proof}
  First, observe that if we set $\hat{x}_{\suffChain}=1$, then we must
  have had that $\bar{x}_{\suffChain}>0$ (or $\suffChain=\canonFam^P$ for some
  $P\in\P$ where $\canonFam$ is the canonical suffix chain family for
  $\I[\late]$); thus, for all other $\suffChain$ we immediately have that
  $\Pr[\hat{x}_{\suffChain}=1]=0=\bar{x}_{\suffChain}$.

  We begin by focusing on suffix chains $\suffChain$ in the support of $\bar{x}$.  Recall that
  the fractional solution $(U,x)$ we began with had cross-free
  support, i.e.~we had that $\{\suffChain\in\S^P:x_{\suffChain}>0\}$ was
  cross-free for all $P\in\P$.  Now, we defined $\bar{x}$ so that
  for all $\suffChain$ we have $\bar{x}_{\suffChain}>0$ implies $x_{\suffChain}>0$.  Thus, we may conclude $(\bar{U},\bar{x})$ has cross-free support as well.  Fix some
  $P\in\P$, and enumerate $\S^P$ in sorted order as
  $\suffChain^1\prec\suffChain^2\prec\dots\prec\suffChain^m$, where
  $m=\abs{\supp(\bar{x})}\le\abs{\S^P}$.  For any $\ell\in[m]$, we can compute
  $\Pr[\hat{x}_{\suffChain^\ell}]$ as follows.  Recall that we set
  $\hat{x}_{\suffChain^\ell}=1$ if and only if we had that
  $\suffChain^\ell=\suffChain^{\pOne}(\alpha)$ or
  $\suffChain^\ell=\suffChain^{\pTwo}(1-\alpha)$ (for $P=\pOne$ or $P=\pTwo$,
  respectively).  Now, from our definition of $\suffChain^P(\cdot)$, we
  can see that these equalities hold, respectively, if and only if we
  have that
  \begin{equation*}
    {
    \sum_{\ell'=1}^{\ell-1}\bar{x}_{\suffChain^{\ell'}}\le\alpha<\sum_{\ell'=1}^{\ell}\bar{x}_{\suffChain^{\ell'}}
    \qquad\text{or}\qquad
    \sum_{\ell'=1}^{\ell-1}\bar{x}_{\suffChain^{\ell'}}\le1-\alpha<\sum_{\ell'=1}^{\ell}\bar{x}_{\suffChain^{\ell'}}.
    }
  \end{equation*}
  Thus, the probability of selecting $\suffChain^\ell$ is precisely the
  probability of choosing an $\alpha$ in one of the ranges above.
  Recalling that $\alpha\sim\mathcal{U}[0,1]$, that
  $\bar{x}_{\suffChain}$ is nonnegative for all $\suffChain$, and that
  $\sum_{\suffChain\in\S^P}\bar{x}_{\suffChain}\le1$ always (by
  Lemma~\ref{lem:boostfeas_corr}), we can conclude this probability
  is, in fact, the lengths of the intervals in question.  Since both
  intervals have the same length, in either case we get that
  \begin{equation*}
    \Pr[\hat{x}_{\suffChain^\ell}=1]
    =
    \sum_{\ell'=1}^{\ell}\bar{x}_{\suffChain^{\ell'}}-\sum_{\ell'=1}^{\ell-1}\bar{x}_{\suffChain^{\ell'}}
    =
    \bar{x}_{\suffChain^\ell},
  \end{equation*}
  exactly as desired.

  Finally, for each $P\in\P$, we consider $\Pr[\hat{x}_{\canonFam^P}=1]$,
  where $\canonFam$ is the canonical suffix chain family for $\I[\late]$.
  Again, by the definition of our rounding process we can see this
  happens precisely when $\canonFam^{\pOne}=\suffChain^{\pOne}(\alpha)$ or
  $\canonFam^{\pTwo}=\suffChain^{\pTwo}(1-\alpha)$.  This two equalities holds,
  respectively, precisely when
  \begin{equation*}
    \sum_{\suffChain\in\S^{\pOne}}^{m}\bar{x}_{\suffChain}\le\alpha
    \qquad\text{or}\qquad
    \sum_{\suffChain\in\S^{\pTwo}}^{m}\bar{x}_{\suffChain}\le1-\alpha.
  \end{equation*}
  As in the previous case, however, the nonnegativity of $\bar{x}$ and
  Lemma~\ref{lem:boostfeas_corr} allow us to conclude that
  \begin{equation*}
    0
    \le
    \sum_{\suffChain\in\S^{\pOne}}^{m}\bar{x}_{\suffChain},
    \sum_{\suffChain\in\S^{\pTwo}}^{m}\bar{x}_{\suffChain}
    \le1.
  \end{equation*}
  Thus, since $\alpha\sim\mathcal{U}[0,1]$, we can conclude that in either case we have that
  \begin{equation*}
    \Pr[\hat{x}_{\canonFam^P}=1]=1-\sum_{\suffChain\in\S^P}\bar{x}_{\suffChain}
  \end{equation*}
  exactly as claimed.
\end{proof}

\begin{lemma}
  \label{lem:chernoff_sub_corr}
  With probability 1, we have that
  \begin{equation}
    \label{eq:boostfeas_corr}
    \sum_{P\in\P}\sum_{\suffChain\in\S^P}p_{\suffChain}^{i,\canonFam}\hat{x}_{\suffChain}\ge\Theta^{i,\canonFam}
  \end{equation}
  for all $i\in[k]$, where $\canonFam$ is the canonical suffix chain family
  for $\I[\late]$.
\end{lemma}
\begin{proof}
  We obtain the desired bound by combining upper and lower bounds for
  the expected value of the sum on the left of inequality~\eqref{eq:boostfeas_corr}:
  $\E[\sum_{P\in\P}\sum_{\suffChain\in\S^P}p_{\suffChain}^{i,\canonFam}\hat{x}_{\suffChain}]$.
  Fix some $i\in[k]$, and let $M$ be the quantity we want to lower
  bound, i.e.~the minimum value that the sum inside the expectation achieves.
  First, we show that the {\em maximum} value this sum ever
  achieves is at most $M+\Theta^{i,\canonFam}$; this immediate implies that
  the expected value of the sum is also at most $M+\Theta^{i,\canonFam}$.  We
  then show a lower bound of $2\Theta^{i,\canonFam}$ on the expected value of
  the sum.  Chaining these two inequalities together immediately
  gives us that $M\ge\Theta^{i,\canonFam}$, exactly as desired.

  We begin by proving our claimed upper bound on the sum
  $\sum_{P\in\P}\sum_{\suffChain\in\S^P}p_{\suffChain}^{i,\canonFam}\hat{x}_{\suffChain}$.
  First, we simplify this sum by recalling details of our rounding
  procedure.  For each path $P\in\P$ we set exactly one
  $\hat{x}_{\suffChain}$ equal to $1$, and set all others equal to $0$,
  based on a uniform random variable $\alpha\sim\mathcal{U}[0,1]$.
  Thus, our sum reduces to precisely the coefficients of the two
  variables in question; by the definition of our rounding procedure,
  we can see that we get that
  \begin{equation*}
    \sum_{P\in\P}\sum_{\suffChain\in\S^P}p_{\suffChain}^{i,\canonFam}\hat{x}_{\suffChain}\ge\Theta^{i,\canonFam}
    =
    p_{\suffChain^1}^{i,\canonFam}+p_{\suffChain^2}^{i,\canonFam},
  \end{equation*}
  where $\suffChain^1=\suffChain^{\pOne}(\alpha)$ and
  $\suffChain^2=\suffChain^{\pTwo}(1-\alpha)$.

  The key to our upper bound is showing that for any two values of $\alpha$, the resulting values of the sum     $p_{\suffChain^1}^{i,\canonFam}+p_{\suffChain^2}^{i,\canonFam}$ can differ by at most $\Theta^{i,\canonFam}$.  To that end, fix
  $0\le\alpha\le\alpha'\le1$; set
  $\suffChain^1=\suffChain^{\pOne}(\alpha)$ and
  $\suffChain^2=\suffChain^{\pTwo}(1-\alpha)$ (as before); and set
  $\suffChain^{\prime1}=\suffChain^{\pOne}(\alpha')$ and
  $\suffChain^{\prime2}=\suffChain^{\pTwo}(1-\alpha')$.  Now, by our definition
  of $\suffChain^{P}(\cdot)$, we can immediately see that
  $\alpha\le\alpha'$ implies that $\suffChain^1\preceq\suffChain^{\prime1}$
  and $\suffChain^2\succeq\suffChain^{\prime2}$, and so $\suffChain^1_i\supseteq\suffChain^{\prime1}_i$
  and $\suffChain^2_i\subseteq\suffChain^{\prime2}_i$.  From the definition of $p_{\suffChain}^{i,\canonFam}$, however, we can then see that we have
  \begin{equation*}
    0 \le p_{\suffChain^{\prime1}}^{i,\canonFam}\le p_{\suffChain^1}^{i,\canonFam} \le \Theta^{i,\canonFam}
    \qquad\text{and}\qquad
    0 \le p_{\suffChain^2}^{i,\canonFam}\le p_{\suffChain^{\prime2}}^{i,\canonFam} \le \Theta^{i,\canonFam}.
  \end{equation*}
  The above immediately imply the slightly weaker pair of inequalities
  \begin{equation*}
    p_{\suffChain^{\prime1}}^{i,\canonFam}\le p_{\suffChain^1}^{i,\canonFam} \le p_{\suffChain^{\prime1}}^{i,\canonFam}+\Theta^{i,\canonFam}
    \qquad\text{and}\qquad
    p_{\suffChain^2}^{i,\canonFam}\le p_{\suffChain^{\prime2}}^{i,\canonFam} \le p_{\suffChain^2}^{i,\canonFam}+\Theta^{i,\canonFam},
  \end{equation*}
  which combine to imply that
  \begin{equation*}
    (p_{\suffChain^1}^{i,\canonFam}+p_{\suffChain^2}^{i,\canonFam}) \le (p_{\suffChain^{\prime1}}^{i,\canonFam}+p_{\suffChain^{\prime2}}^{i,\canonFam})+\Theta^{i,\canonFam}
    \qquad\text{and}\qquad
    (p_{\suffChain^{\prime1}}^{i,\canonFam}+p_{\suffChain^{\prime2}}^{i,\canonFam})\le(p_{\suffChain^1}^{i,\canonFam}+ p_{\suffChain^2}^{i,\canonFam})+\Theta^{i,\canonFam}.
  \end{equation*}
  Recall that we chose
  $M=\min_{\alpha}\sum_{P\in\P}\sum_{\suffChain\in\S^P}p_{\suffChain}^{i,\canonFam}\hat{x}_{\suffChain}=\min_\alpha(p_{\suffChain^1}^{i,\canonFam}+p_{\suffChain^2}^{i,\canonFam})$.
  Thus, if we let $\alpha^\ast$ achieve this minimum value $M$, i.e.~pick $\alpha^\ast\in\argmin_\alpha(p_{\suffChain^1}^{i,\canonFam}+p_{\suffChain^2}^{i,\canonFam})$, we can
  see that applying the two inequalities above in the regions
  $[0,\alpha^\ast]$ and $[\alpha^\ast,1]$, respectively, immediately
  gives us that
  \begin{equation*}
    \E[\sum_{P\in\P}\sum_{\suffChain\in\S^P}p_{\suffChain}^{i,\canonFam}\hat{x}_{\suffChain}]
    =\E[p_{\suffChain^1}^{i,\canonFam}+p_{\suffChain^2}^{i,\canonFam}]
    \le\E[M+\Theta^{i,\canonFam}]
    =M+\Theta^{i,\canonFam},
  \end{equation*}
  exactly as claimed.

  Our lower bound follows simply by combining Lemmas~\ref{lem:exp_cost_corr}
  and~\ref{lem:boostfeas_corr}.  In particular, we get that
  \begin{equation*}
    \E[\sum_{P\in\P}\sum_{\suffChain\in\S^P}p_{\suffChain}^{i,\canonFam}\hat{x}_{\suffChain}]
    =
    \sum_{P\in\P}\sum_{\suffChain\in\S^P}p_{\suffChain}^{i,\canonFam}\Pr[\hat{x}_{\suffChain}=1]
    =
    \sum_{P\in\P}\sum_{\suffChain\in\S^P}p_{\suffChain}^{i,\canonFam}\bar{x}_{\suffChain}
    \ge
    2\Theta^{i,\canonFam},
  \end{equation*}
  where the first equality follows since the
  $\hat{x}_{\suffChain}$ are all binary random variables, the second follows by
  Lemma~\ref{lem:exp_cost_corr}, and the inequality follows by Lemma~\ref{lem:boostfeas_corr}.
  Combining this with our previously found upper bound on this
  expected value, however, we can conclude that
  $M+\Theta^{i,\canonFam}\ge2\Theta^{i,\canonFam}$ or equivalently $M\ge\Theta^{i,\canonFam}$,
  exactly as desired.
\end{proof}

Finally, we observe that we can use the same procedure as outlined in
Section~\ref{sec:ex_ip_rounding} to produce a feasible schedule for $\I$ from the
integral solution $(\hat{U},\hat{x})$. As before, on each
path $P\in\P$ we defer all jobs that are either in the suffix
chain $\suffChain\in\S^P$ such that $\hat{x}_{\suffChain}=1$ {\em or} in the
canonical suffix for $P$ in $\I[\late]$, i.e.~all jobs in the join
$\suffChain\wedge \canonFam^P$, and then running the algorithm outlined in that
section to get a feasible schedule.  We get the following results
corresponding to Theorem~\ref{thm:sched_good} and Corollary~\ref{cor:sched_cost}; since their proofs
follow from largely the same arguments as those given in
Section~\ref{sec:ex_ip_rounding}, we only sketch the differences below.

\begin{theorem}
\label{thm:sched_good_corr}
  Applying the process and algorithm from Section~\ref{sec:ex_ip_rounding} to the
  solution $(\hat{U},\hat{x})$ yields a feasible schedule with cost at
  most $\sum_{j\notin\late}w_j\hat{U}_j$ in the instance $\I[\late]$.
\end{theorem}
\begin{proof}
  The proof of this theorem follows from the exact same argument as
  that for Theorem~\ref{thm:sched_good}.  The only difference is the
  equivalent of constraint~\eqref{eq:feas} for our current rounding
  scheme holds with certainty, due to
  Lemma~\ref{lem:chernoff_sub_corr}.  Specifically, we have that
  \begin{equation*}
    \sum_{P\in\P}\sum_{\suffChain\in\S^P}p_{\suffChain}^{i,\canonFam}\hat{x}_{\suffChain}\ge\Theta^{i,\canonFam},
  \end{equation*}
  holds with probability $1$, and hence our upper bound on costs holds
  unconditionally, unlike its counterpart from
  Section~\ref{sec:ex_ip_rounding}.  In all other aspects the theorem
  and its proof are identical to those found in
  Section~\ref{sec:ex_ip_rounding}, and we refer the reader to that section for further details.
\end{proof}

\begin{corollary}
  \label{cor:sched_cost_corr}
  The schedule produced by following the process and algorithm from
  Section~\ref{sec:ex_ip_rounding} is feasible in the original instance $\I$ of
  $\pdls$, and incurs expected penalty of $2\sum_jw_jU_j$.
\end{corollary}
\begin{proof}
  The proof largely follows the same arguments as that of
  Corollary~\ref{cor:sched_cost}, and we refer the reader to that proof for full
  details.  We sketch the main technical details below.  Similarly to
  that section, since we defined $\late=\{j:U_j\ge1/2\}$, and set $\bar{U}_j=2U_j$ for all $j\notin\late$, we get that
  \begin{equation*}
    \sum_{j\notin\late}w_j\bar{U}_j+\sum_{j\in\late}w_j\le2\sum_jw_jU_j.
  \end{equation*}
  Now, from Theorem~\ref{thm:sched_good_corr} we have the schedule we produce in the end
  has cost at most $\sum_{j\notin\late}w_j\hat{U}_j$ in the modified
  setting $\I[\late]$.  Recall, however, that the only difference
  between $\I$ an $\I[\late]$ is that in the latter we increased the
  deadlines of all jobs in $\late$ to be $\Gamma$.  Thus, the cost of
  our schedule in $\I$ can increase versus the cost in $\I[\late]$ by
  at most the total penalty of all jobs in $\late$, i.e.~by at most
  $\sum_{i\in\late}w_j$.  Applying Lemma~\ref{lem:exp_cost_corr}, we can thus see that the expected cost of our schedule in $\I[\late]$ is at most
  \begin{equation*}
    \E[\sum_{j\notin\late}w_j\hat{U}_j+\sum_{i\in\late}w_j]
    =\sum_{j\notin\late}w_j\bar{U}_j+\sum_{j\in\late}w_j
    \le2\sum_jw_jU_j,
  \end{equation*}
  exactly as claimed.
\end{proof}

  \section{LP gap example}
  \label{sec:lp_gap}
  
In this section, we show that using the KC-inequalities~\eqref{c2} in the LP relaxation~\eqref{p} is critical to
our approximation factor.  In particular, we show that without
utilizing KC-inequalities, the integrality gap of (P) would be
$\Omega(n/\log n)$.  We do so by constructing a simple instance $\I$
of the \pdls, and considering the program (P) without the
KC-inequality strengthening of constraint~\eqref{c2}.  We will show
that the LP relaxation of this weaker version of (P) admits a
fractional solution whose value is a factor of $\Omega(n/\log k)$
better than any integral solution.  Of special note is the simplicity
of the instance $\I$ we construct: the instance consists of a single
path, all of the jobs on which have unit process time and unit weight.

In the rest of the section, we describe the instance $\I$ of \pdls and
the fractional solution $(U,x)$ yielding our claimed integrality gap.
As mentioned, our example consists of a single path
$P=1\preceq2\preceq\dots\preceq n$, with $w_j=p_j=1$ for all
$j\in[n]$.  Each job has one of $k=n/2$ distinct deadlines, with the
deadline of job $j$ being given by
\begin{equation*}
  d_j
  =
  \begin{cases}
    j &\text{if $j$ is odd; and}\\
    j-1&\text{if $j$ is even.}
  \end{cases}
\end{equation*}
Thus, we can see that deadline $D_i=2i-1$ for all $i\in[k]$, and that
$i(j)=\lfloor\frac{j+1}{2}\rfloor$.  

Since our example consists of a single path, there is only
one feasible solution: simply schedule the jobs in the order they
appear on that path.  We chose our deadlines such that this schedule runs jobs
with odd indices  on time, and make jobs with even indices
late.  This results in a total cost of $n/2$.

On the other hand, consider the following fractional solution to
program~(P) for the given example.\footnote{in the interests of
  simplicity, for the rest of this discussion we omit any uses of the
  index $P$ and use $\le$ in place of $\preceq$ since, by
  construction, they are equivalent.}  The support of our solution
will be the set $\{\suffChain^\ell\}_{\ell}$ of suffix chains indexed by
$\ell\in[k+1]$, where we define $\suffChain^\ell$ by
\begin{equation*}
  \suffChain^\ell_i=
\begin{cases}
  \{1,2,\dots,n\}&\qquad\text{if }i\le\ell\text{; and}\\
  \{2i+1,2i+2,\dots,n\}&\qquad\text{if }i>\ell.
\end{cases}
\end{equation*}
Our proposed solution then sets
\begin{equation*}
  x_{\suffChain}=
\begin{cases}
  \frac{1}{2}&\qquad\text{if }\suffChain=\suffChain^0\text{; }\\
  \frac{1}{2\ell}-\frac{1}{2(\ell+1)}&\qquad\text{if }\suffChain=\suffChain^\ell\text{ for $1\le\ell<k$; }\\
  \frac{1}{2k}&\qquad\text{if }\suffChain=\suffChain^k\text{; and}\\
  0\qquad\text{otherwise}.
  \end{cases}
\end{equation*}
Note that this solution has far lower cost than the integral one.  In
particular, for any $j\in[n]$, the suffix chain $\suffChain^\ell$ makes job $j$ late if and only if $j\in \suffChain^\ell_{i(j)}$.  Whenever $i(j)\le\ell$, we have $j\in \suffChain^\ell_j=[n]$; on the other hand, for $i(j)>\ell$ we can see that the smallest element of 
 $\suffChain^{\ell}_{i(j)}$ is 
$2\lfloor\frac{j+1}{2}\rfloor+1 
  \ge j+1$,
since  $i(j)=\lfloor\frac{j+1}{2}\rfloor$.  
So  $j\in \suffChain^\ell_{i(j)}$ if and only if $i(j)\le\ell$. 
Thus, we have that 
\begin{equation*}
  U_j
  =
  \sum_{\suffChain\in\S:\suffChain_{i(j)}\ni j}x_{\suffChain}
  =
  \sum_{\ell=i(j)}^{k}x_{\suffChain^\ell}
  =\sum_{\ell=i(j)}^{k-1}\left(\frac{1}{2\ell}-\frac{1}{2(\ell+1)}\right) + \frac{1}{2k}
  \le\frac{1}{j}.
  \end{equation*}
The first inequality above follows by the definition of $x_{\suffChain}$; the
second by our definition of the support of $x$ and our observations
above on $\suffChain^\ell_{i(j)}$; the third by our choice of $x$; and
the last by observing that we have a telescoping sum, and recalling that we have  $2i(j)=2\lfloor\frac{j+1}{2}\rfloor\ge j$.  Thus, we can see that the solution $(U,x)$ produces objective value 
\begin{equation*}
  \sum_{j\in[n]}w_jU_j
  \le
  \sum_{j\in[n]}\frac{1}{j}
  =O(\log n).
\end{equation*}
Thus, any program with integrality gap better than $\Omega(n/\log n)$
must have constraints which this candidate solution $(U,x)$ violates.
We next show that $(U,x)$ satisfies constraint~(C1) from program~(P),
as well as the version of constraint~(C2) which does not use the
KC-inequalities strengthening.  This implies that the KC-inequalities
are critical to achieving a small LP gap.

We begin by showing that $(U,x)$ satisfies constraint~(C1) from the
program (P).  First, note that we have defined $x$ so that
\begin{equation*}
  \sum_{\suffChain\in\S}x_{\suffChain}
  =\sum_{\ell=0}^{\ell=k} x_{\suffChain^\ell}
  =\frac{1}{2}+\sum_{\ell=1}^{\ell=k-1}\left(\frac{1}{2\ell}-\frac{1}{2(\ell+1)}\right) + \frac{1}{2k}
  =1,
\end{equation*}
and so constraint~(C1) is satisfied.  

Since our fractional solution $(U,x)$
satisfies~(C1), we conclude constraint~(C2) is crucial to
bounding the integrality gap of (P).  Furthermore, we will now show
that the ``capping'' of processing times $p_{\suffChain}^{i,\suffChainFam}$ by $\Theta^{i,\suffChainFam}$
is critical to achieving a good integrality gap for (P).  
We do so by demonstrating that {\em without} the capping operation,
 $(U,x)$ would, in fact, satisfy
constraint~(C2), and hence show the integrality gap of (P) is $\Omega(n/\log
n)$.

Consider constraint~(C2) without the capping operation.
Fix some suffix chain $\suffChainFam$ on $\P$. Using the previously given
definition, for each $i\in[k]$ we have that
\begin{equation*}
  \Theta^{i,\suffChainFam}=\max\{n-D_i-\abs{\suffChainFam_i},0\}.
\end{equation*}
Furthermore, if we no longer enforce that the processing times we associate with suffixes be at most
$\Theta^{i,\suffChainFam}$ we get that
\begin{equation*}
  p_{\suffChain}^{i,\suffChainFam}
  =\sum_{j\in\suffChain_i\setminus \suffChainFam_i}p_{j}
  =\abs{\suffChain_i\setminus \suffChainFam_i}.
\end{equation*}
Fix $i\in[k]$.  For constraint~(C2) to hold in our setting,
we need that
\begin{equation*}
  \sum_{\suffChain\in\S:\suffChain_i\supsetneq \suffChainFam_i}p_{\suffChain}^{i,\suffChainFam}x_{\suffChain} \ge \Theta^{i,\suffChainFam}.
\end{equation*}
We now show that the above always holds; we break our proof into three cases.
\begin{itemize}
\item Case: $\abs{\suffChainFam_i}\ge n-D_i$. Then we have that $\Theta^{i,\suffChainFam}=0$,
  and the inequality holds trivially.  
\item Case: $\abs{\suffChainFam_i}= n-D_i-1$.  Then we have that
  $\Theta^{i,\suffChainFam}=1$; furthermore, since $D_i=2i-1$, we can see that
  $\abs{\suffChainFam_i}=n-2i$, and so $\suffChainFam_i=\{2i+1,2i+2,\dots,n\}$.  Recalling the
  definition of $\suffChain^\ell_i$, however, we can see that
  $\suffChain^\ell_i\supsetneq \suffChainFam_i$ only if $i\le\ell$, and so
  $\suffChain^\ell_i=[n]$.  Thus, we get that
  \begin{equation*}
    \sum_{\suffChain\in\S:\suffChain_i\supsetneq \suffChainFam_i}p_{\suffChain}^{i,\suffChainFam}x_{\suffChain}
    =\sum_{\ell=i}^{k}p_{\suffChain^\ell}^{i,\suffChainFam}x_{\suffChain^\ell}
    =\abs{[n]\setminus \suffChainFam_i}\left(\sum_{\ell=i}^{k-1}\left(\frac{1}{2\ell}-\frac{1}{2(\ell+1)}\right)+\frac{1}{2k}\right)
    =2i\cdot\frac{1}{2i}
    =1.
  \end{equation*}
  Thus the desired inequality holds in this case as well.
\item Case: $\abs{\suffChainFam_i}\le n-D_i-2$.  Then we have that $\suffChain^\ell\supsetneq \suffChainFam_i$ for all $0\le\ell\le k$.
    So can see that 
  \begin{align*}
    \sum_{\suffChain\in\S:\suffChain_i\supsetneq \suffChainFam_i}p_{\suffChain}^{i,\suffChainFam}x_{\suffChain}
    &=
    \sum_{\ell=0}^{i-1}(n-2i-\abs{\suffChainFam_i})x_{\suffChain^\ell}
    +
    \sum_{\ell=i}^{k}(n-\abs{\suffChainFam_i})x_{\suffChain^\ell}\\
    &=(n-\abs{\suffChainFam_i})\sum_{\ell=0}^{k}x_{\suffChain^\ell} -2i\sum_{\ell=0}^{i-1}x_{\suffChain^\ell}\\
    &=(n-\abs{\suffChainFam_i}) - 2i\left(1-\frac{1}{2i}\right)\\
    &=(n-\abs{\suffChainFam_i}) - (2i-1)\\
    &=n - D_i - \abs{\suffChainFam_i}\\
    &=\Theta^{i,\suffChainFam},
  \end{align*}
  exactly as required.
\end{itemize}
In every case, we get that constraint~(C2) holds. Since our choice of
$\suffChainFam$ was arbitrary, we may conclude that requiring
$p_{\suffChain}^{i,\suffChainFam}\le\Theta^{i,\suffChainFam}$is critical to ensuring a good
integrality gap for~(P).

  \section{Tightness of rounding scheme}
  \label{sec:example2}
  \newcommand{\canonFamM}{\widetilde{\canonFam}}
\newcommand{\Part}{\Pi}
\newcommand{\eps}{\varepsilon}
\newcommand{\floor}[1]{\lfloor{#1}\rfloor}
\newcommand{\ceil}[1]{\lceil{#1}\rceil}

Here, we give an example to show the limits of our current techniques
when implemented with independent rounding.  We construct an instance
$\I$ and a fractional solution $(U,x)$ for $\I$ such that if we use
the rounding procedure outlined in Section~4, while setting
$\gamma=O(1/\log^\eps k)$ for some $\eps>0$, then our probability of
success will be $o(1)$ in $k$.  In other words, if we want our
rounding procedure to succeed with constant probability, we simply
cannot replace our boosting factor of $\gamma \log k$ with one that is
$O(\log^{1-\eps}k)$.  This shows that the result of
Lemma~\ref{lem:chernoff} is tight, and so the approximation factor in
Corollary~\ref{cor:sched_cost} cannot be improved without significant
new techniques.

We begin by constructing the instance $\I$ of \pdls.  Fix the number
of deadlines $k$, and let $n$ be the number of jobs for some $n$
divisible by $k^2$.  Each of our jobs will have unit runtime and unit
weight, i.e.~$w_j=p_j=1$ for all $j\in[n]$.  Our set of deadlines
will be $\D =\{D_1,\dots,D_k\}$ where
\begin{equation*}
D_i=\left(\frac{i}{k}\right)n-1
\qquad\text{for all } i\in[k].
\end{equation*}
Our set of paths $\P$ contains $n/k$ identical paths of length $k$;
each path contains a single job with deadline $D_i$ for each
$i\in[k]$, in increasing order along the path.  Specifically, if one
of our paths is $P=\{j_1,j_2,\dots,j_k\}$, with $j_1\prec\dots\prec
j_k$, then for all $\ell\in[k]$ we have that $i(j_\ell)=\ell$.  In the
following, we always index the jobs in a path $P$ in this fashion for
convenience, so that for any such path we always have that $j_i\prec
j_{i'}$ if and only if $i<i'$, and that the deadline of job $j_i$ is
$D_i$ for all $i\in[k]$.

We now describe a fractional solution $(U,x)$ for $\I$, parameterized
by $\gamma$.  For every path $P\in\P$, we
will have only two suffix chains of $P$ in the support of $x$.  One
will be the canonical suffix chain $\canonFam^P$ for $P$, and the
other will be a slight modification of the canonical suffix chain,
which we denote $\canonFamM^P$.  Before defining these two suffix
chains formally, we first note that we will place the majority of our
solution's weight on $\canonFam^P$, setting
\begin{equation*}
  x_\suffChain
  =
  \begin{cases}
    1-\frac{1}{2\gamma\log k}&\qquad\text{if }\suffChain=\canonFam^P\text{;}\\
    \frac{1}{2\gamma\log k}&\qquad\text{if }\suffChain=\canonFamM^P\text{; and}\\
    0&\qquad\text{otherwise}.
  \end{cases}
\end{equation*}
Since by definition the canonical suffix chain family never makes {\em
  any} job late, this solution ensures that $U_j\le1/2\gamma\log k$ for all
$j\in[n]$.  Thus, we can immediately conclude that for the solution
$(U,x)$, 
we have
\begin{equation*}
  \late=\{j\in[n]:U_j\ge1/\gamma\log k\}=\emptyset,
\end{equation*}
and so $\I[\late]=\I$.  Thus, our rounding procedure works solely
with the original instance $\I$, and the canonical suffix chain family
$\canonFam$ for $\I$.

We now formally define the canonical suffix chain family $\canonFam$,
the modification $\canonFamM$, and the KC-Inequalities for $\I$
corresponding to $\canonFam$.  First, by inspection we can see that
for any path $P=\{j_1,j_2,\dots,j_k\}\in\P$, we have that
\begin{equation*}
  \canonFam^P=\{j_2,\cdots,j_k\}\supseteq\{j_3,\dots,j_k\}\supseteq\dots\supseteq\{j_k\}\supseteq\emptyset,
\end{equation*}
i.e.~we have that $\canonFam^P_i=\{j_{i+1},j_{i+2},\dots,j_k\}$ for
all $i\in[k]$.  Now, we define the modification $\canonFamM$ as
follows.  First, we partition the set $\P$ of paths into $k$
groups $\Part_1,\Part_2,\dots,\Part_{k}$, each containing exactly $n/k^2$
paths.  Then, for any path $P\in\P$, we define the suffix chain
$\canonFamM^P$ as
\begin{equation*}
  \canonFamM^P_i
  =
  \begin{cases}
    \canonFam^P_i&\qquad\text{if }P\notin\Part_i\text{; and}\\
    \canonFam^P_{i-1}&\qquad\text{if }P\in\Part_i\text{,}
  \end{cases}
\end{equation*}
where we take $\canonFam^P_0$ to indicate the entire chain $P$.
From the definitions of $\canonFam$, $\canonFamM$, and $x$
given above, we can compute that for each path
$P=\{j_1,\dots,j_k\}\in\P$, and each $i\in[k]$, we have that
\begin{equation*}
  U_{j_i}
  =
  \begin{cases}
    \frac{1}{2\gamma\log k}&\qquad\text{if }P\in\Part_i\text{; and}\\
    0&\qquad\text{otherwise.}
  \end{cases}
\end{equation*}
Thus, as previously mentioned, we can see that no job is
made late to an extent of $1/\gamma\log k$ or more, and so $\I[\late]=\I$.

\begin{lemma}
  \label{lem:okToRound}
  For the described instance $\I$ of \pdls, the constructed solution
  $(U,x)$ satisfies conditions (a) and (b) of the rounding procedure
  of Section~\ref{sec:ex_ip_rounding}, i.e.~the solution has
  cross-free support and satisfies the reduced constraint
  $(\text{C2'})$, whenever we have $n/k^2\ge2\gamma\log k$, where $\gamma$
  is the parameter for the rounding process.
\end{lemma}
\begin{proof}
  We begin by showing that condition (a) is satisfied, i.e.~the
  constructed $x$ has cross-free support.  Fix some $P\in\P$.  Now, we
  defined $x$ such that $x_{\suffChain}>0$ if and only if
  $\suffChain\in\{\canonFam^P,\canonFamM^P\}$ for all
  $\suffChain\in\S^P$.  Recall, however, that for all $i$ we either
  have that $\canonFamM_i^P=\canonFam_i^P$, or have that
  $\canonFamM_i^P=\canonFam_{i-1}^P\supset\canonFam_i^P$.  Thus, we
  may conclude that $\canonFamM^P\preceq\canonFam^P$ for all $P\in\P$,
  and hence the support of $x$ is cross-free.

  Now, we show that condition (b) is satisfied, i.e.~the
  KC-Inequalities corresponding to the canonical suffix chain family
  $\canonFam$ for the modified \pdls instance $\I[\late]$ are
  satisfied.  Recall, however, that we already saw that
  $U_j<1/\gamma\log k$ for all $j\in[n]$, and so $\I[\late]=\I$;
  thus, we are actually interested in the KC-Inequalities
  corresponding to the canonical suffix chain family for the {\em original}
  \pdls instance $\I$.

  We begin by calculating the relevant constants for the
  KC-Inequalities.  First, we compute $\Theta^{i,\canonFam}$. Recall
  that all of our jobs had unit processing times, and so we have
  $p_j=1$ for all $j\in[n]$ and $\Gamma=n$.  We claim that this gives
  us that $\Theta^{i,\canonFam}=1$ for all $i\in[k]$.  Fix some
  $i\in[k]$. Now, note that all of our paths are identical, and for
  each path $P=\{j_1,\dots,j_k\}\in\P$ we have that
$
  \abs{\canonFam^P_i}
  =
  \abs{\{j_{i+1},\dots,j_k\}}
  =
  k-i.
$ 
Thus, we can see that
\begin{equation*}
  \Theta^{i,\canonFam}
  =
  (\Gamma-D_i)-\abs{P}(k-i)
  =
  \left(n - \frac{i}{k}\cdot n+1\right)-\frac{n}{k}\left(\vphantom{\frac{i}{k}}k-i\right)
  =
  1.
\end{equation*}

Now, we consider the values of $p_{\canonFam^P}^{i,\canonFam}$ and
$p_{\canonFamM^P}^{i,\canonFam}$.  First, we note that since 
$p_{\suffChain}^{i,\canonFam}$ denotes the (possibly capped) number of jobs suffix chain $\suffChain$ defers {\em
  in addition} to those deferred by the canonical suffix chain family
$\canonFam$, we immediately can see that
$p_{\canonFam^P}^{i,\canonFam}=0$ always.  Furthermore, since for each
path $P\in\P$, $\canonFamM$ differs from $\canonFam$ only in that it
defers a single additional job past deadline $D_\ell$ where
$P\in\Part_\ell$, we conclude that
\begin{equation*}
  p_{\canonFamM^P}^{i,\canonFam}
  =
  \begin{cases}
    1&\qquad\text{if }P\in\Part_i\text{; and}\\
    0&\qquad\text{otherwise}.
  \end{cases}
\end{equation*}

Combining the above, we can see that the KC-Inequality corresponding
to the canonical suffix chain family $\canonFam$ holds whenever
$n/k^2\ge2\gamma\log k$, exactly as claimed.  To see this, we first compute
that, for any $i\in[k]$ we have
\begin{equation*}
  \sum_{P\in\P}\sum_{\suffChain\in\S^P}p_{\suffChain}^{i,\canonFam}x_{\suffChain}
=
\sum_{P\in\P}p_{\canonFamM^P}^{i,\canonFam}\frac{1}{2\gamma\log k}
=\abs{\Part_i}\frac{1}{2\gamma\log k}
=\frac{n}{k^2}\cdot\frac{1}{2\gamma\log k},
\end{equation*}
where the first inequality follows by recalling that
$x_{\suffChain}>0$ only when
$\suffChain\in\{\canonFam^P,\canonFamM^P\}$, and
$p_{\canonFam^P}^{i,\canonFam}=0$; the second follows since
$p_{\canonFamM^P}^{i,\canonFam}=1$ if and only if $P\in\Part_i$ and is
$0$ otherwise; and the last since $\abs{\Part_i}=n/k^2$.  Thus, since
we already saw that $\Theta^{i,\canonFam}=1$, we can substitute in our
computed values and rearrange terms to get that the $(U,x)$ satisfies
the KC-Inequalities corresponding to the canonical suffix chain family
$\canonFam$ if and only if $n/k^2\ge2\gamma\log k$.
\end{proof}
\begin{lemma}
  Let $\gamma$ be a function of $k$, such that
  $\gamma=O(1/\log^{\eps}k)$ for some $\eps>0$.  For the instance $\I$
  of \pdls described above, with $n=k^2\ceil{2\gamma\log k}$,
  applying the rounding procedure of Section~\ref{sec:ex_ip_rounding}
  to $(U,x)$ using $\gamma\log k$ as the boosting parameter has success
  probability that is $o(1)$ in $k$.
\end{lemma}
\begin{proof}
  We begin by noting that the parameter settings outlined above are
  consistent with our example so far; in particular, we chose $n$ to
  be divisible by $k^2$, and furthermore such that
  $n/k^2=\ceil{2\gamma\log k}$ and so the condition for
  Lemma~\ref{lem:okToRound} holds.  

  In order to calculate the probability of the rounding process
  succeeding, we begin by describing the conditions for it to succeed.
  Our claim essentially states that Lemma~\ref{lem:chernoff} from
  Section~\ref{sec:ex_ip_rounding} is tight, and we build on the
  analysis used to prove that lemma.  Considering that lemma, we see
  that for each $i\in[k]$, one random variable $X_P$ is defined for
  each path $P\in\P$ as
  \begin{equation*}
    X_P:=\sum_{\suffChain\in\S^P}p_{\suffChain}^{i,\canonFam}\hat{x}_{\suffChain};
  \end{equation*}
  $\hat{x}_{\suffChain}$ is a random variable obtained by first
  modifying the solution $(U,x)$ to produce a new solution
  $(\bar{U},\bar{x})$, and then using the value of $\bar{x}$ to define
  the marginal distribution for $\hat{x}$.  While we refer the reader
  to Section~\ref{sec:ex_ip_rounding} for the full details, we briefly
  describe the results of this process for our specific solution
  $(U,x)$ in the instance $\I$.  First, we note that from the
  definition of $(U,x)$ in this section and the method for producing
  $\bar{x}$, we will have that
  $\bar{x}_{\canonFamM^P}=(\gamma\log k)\cdot(1/2\gamma\log k)=1/2$ and
  $\bar{x}_{\suffChain}=0$ for all other
  $\suffChain\in\S^P\setminus\{\canonFamM^P\}$.  Second, we note that
  $p_{\canonFamM^P}^{i,\canonFam}=1$ if $P\in\Part_i$ and equals $0$
  otherwise for all $P\in\P$ (see the proof of Lemma~\ref{lem:okToRound} for
  details).  Thus, if we consider the rounding process used to produce
  $\hat{x}$, we will see that for each $P\in\Part_i$, $X_P$ will be a
  binary random variable which takes values $0$ and $1$ with equal
  probability; and for each $P\in\P\setminus\Part_i$, $X_P=0$ always.
  Furthermore, the random variables for each $P\in\P$ are independent.
  
  Now, the rounding procedure succeeds for deadline $D_i$ if and only
  if the sum $X=\sum_{P\in\P}X_P$ of the above variables is at least
  $\Theta^{i,\canonFam}$.  Now, for the setting $\I$ we have that
  $\Theta^{i,\canonFam}=1$ (see the proof of Lemma~\ref{lem:okToRound} for
  details).  As we saw above, however, $X_P$ is identically $0$
  whenever $P\in\P\setminus\Part_i$, so we conclude that the rounding
  process succeeds for deadline $D_i$ if and only if
  \begin{equation*}
    X=\sum_{P\in\Part_i}X_P\ge\Theta^{i,\canonFam}=1.
  \end{equation*}
  Recalling that each of the $X_P$ above is independently $1$ with
  probability $1/2$ and $0$ otherwise, we can see that
  \begin{equation*}
    \Pr[X\ge1] 
    = 
    1 - \Pr[X=0] 
    = 
    1 - \prod_{P\in\Part_i}\Pr[X_P=0]
    =
    1 - \left(\frac{1}{2}\right)^{\abs{\Part_i}}
    =
    1 - \frac{1}{2^{n/k^2}},
  \end{equation*}
  and so the probability our rounding scheme succeeds for deadline $D_i$
  is precisely $1-2^{-n/k^2}$.

  Now, given the above, we want to compute the overall probability
  that our rounding scheme succeeds.  In fact, this probability is
  precisely the probability that our rounding scheme succeeds for all
  of the deadlines in $\D$.  Above, we saw that the probability our
  scheme succeeded for any single deadline was $1-2^{-n/k^2}$; we must
  be careful, however, because while our rounding scheme is
  independent for each path $P\in\P$, there is dependence in its
  behavior for a given path $P\in\P$ with respect to the different
  deadlines $D_i$.  The key observation, however, is that the
  probability we succeed for deadline $D_i$ {\em only} depends on the
  random variables associated with paths $P\in\Part_i$.  Since
  $\Part_1\dots,\Part_k$ partition the set $\P$ of paths 
  and the variables associated with each path $P\in\P$ are
  independent for {\em every} deadline, we have ensured  by
  construction  that the events that we succeed with
  respect to deadlines $D_i$ and $D_{i'}$ will be independent whenever
  $i\neq i'$.  Thus, we conclude that the probability our
  rounding procedure succeeds for $(U,x)$ in the instance $\I$ is
  precisely
  \begin{equation*}
    \Pr[X\ge1\text{ for $D_i$ for all }i\in[k]]
    =
    \prod_{i\in[k]}\Pr[X\ge1\text{ for }D_i]
    =
    \left(1-\frac{1}{2^{n/k^2}}\right)^k.
  \end{equation*}

  Finally, we show that the probability we computed above is $o(1)$
  when we have that $\gamma$ satisfies
  $\gamma=\Omega(1/\log^{\eps}(k))$, some $\eps>0$, and
  $n=k^2\ceil{2\gamma\log k}$.  First, we note that since
  $\gamma=O(1/\gamma\log^{\eps} k)$, we must have that as $k$ goes to
  infinity, either $\gamma\log k$ either converges to some
  fixed constant $c\ge0$, or goes to infinity as well.

  First, if $\gamma\log k\rightarrow c$, some $c\ge0$, as
  $k\rightarrow\infty$, we immediately get that the success
  probability above is $o(1)$.  This is because we can see that for
  all sufficiently large $k$, we have that $n/k^2=\ceil{2\gamma\log
    k}\le2c+1$.  Thus, we can bound our probability of success as
  \begin{equation*}
    \left(1-\frac{1}{2^{n/k^2}}\right)^k\le\left(1-\frac{1}{2^{2c+1}}\right)^{k};
  \end{equation*}
  since $c>0$ is a constant independent of $k$, we may conclude that
  our success probability goes to $0$ as $k$ goes to infinity.

  Second, we show that if $\gamma\log k\rightarrow\infty$ as
  $k\rightarrow\infty$, we again get that our success probability is
  $o(1)$.  First, we note that this implies we must have that
  $n/k^2\rightarrow\infty$ as $k\rightarrow\infty$; thus, we may
  conclude that for all sufficiently large $k$, we have that
  \begin{equation}
    \label{eq:eApproxBd}
    \left(1-\frac{1}{2^{n/k^2}}\right)^k\le\left(\frac{2}{e}\right)^{\frac{k}{2^{n/k^2}}}.
  \end{equation}
  Thus, if we want to show that our success probability is $o(1)$, we
  need only show that the fraction $k/2^{n/k^2}\rightarrow\infty$ as
  $k\rightarrow\infty$.  Note, however, that since
  $\gamma=O(1/\log^{\eps}k)$, we know  (again, for sufficiently
  large $k$) that $n/k^2=\ceil{2\gamma\log k}<\lg^{1-\eps/2}k$, where
  $\lg k$ is the logarithm base $2$ of $k$.  Thus, we can conclude
  that when $k$ is sufficiently large we have that
  \begin{equation}
    \label{eq:exponentBd}
    \frac{k}{2^{n/k^2}}
    \ge
    \frac{k}{2^{\lg^{1-\eps/2}}k}
    =
    k^{\left(1-\frac{1}{\lg^{\eps/2}k}\right)}.
  \end{equation}
  Now, since $\frac{1}{\lg^{\eps/2}k}\rightarrow0$ as $k\rightarrow\infty$, we
  may conclude from equation~\eqref{eq:exponentBd} that
  $\frac{k}{2^{n/k^2}}\rightarrow\infty$ as $k\rightarrow\infty$.  Combining
  this with equation~\eqref{eq:eApproxBd}, however, we can see that
  our success probability must converge to $0$ as $k$ goes to
  infinity. 

  Thus, as we have shown above, in either case we get that our success
  probability goes to $0$ as $k$ goes to infinity, i.e.~our rounding
  process succeeds with probability $o(1)$ in $k$ exactly as claimed.
\end{proof}

}

\end{document}